%% file: paper.tex
\documentclass[runningheads]{llncs}
\pdfoutput=1
\usepackage[utf8]{inputenc}
\usepackage{amsmath}
\usepackage{amssymb}
\usepackage{mathpartir}
\usepackage{hyperref}
\usepackage{listings}
\lstdefinelanguage{michelson}{
  basicstyle=\fontsize{8}{9.6}\selectfont,
  morekeywords={parameter,storage,or,unit,mutez,pair,bool,address}, sensitive=false,
  morecomment=[l]{\#},
  morecomment=[s]{/*}{*/},
  morestring=[b]",
}
\input{macros}

\begin{document}
\title{A Typed Programmatic Interface to Contracts on the Blockchain}
%
%
\author{Thi Thu Ha Doan\orcidID{0000-0001-7524-4497}\and
  Peter Thiemann\orcidID{0000-0002-9000-1239}}

\authorrunning{Ha Doan, P. Thiemann}
%
\institute{University of Freiburg, Germany \\
  \email{\{doanha,thiemann\}@informatik.uni-freiburg.de}
}
\maketitle              
\begin{abstract}
  Smart contract applications on the blockchain can only reach their full potential if
  they integrate seamlessly with traditional software systems via a
  programmatic interface. This interface should provide for
  originating and invoking contracts as well as observing the state of
  the blockchain. We propose a typed API for this 
  purpose and establish some properties of the combined
  system. Specifically, we provide an execution model that
  enables us to prove type-safe interaction between programs and the 
  blockchain. We establish further properties of the model that
  give rise to requirements on the API. A prototype of the interface
  is implemented in OCaml for the Tezos blockchain.

\end{abstract}

\section{Introduction}
\label{sec:introduction}
First generation blockchains were primarily geared towards supporting
cryptocurrencies. Bitcoin is the most prominent system of this kind
\cite{bitcoin-whitepaper}. Although Bitcoin already features a
rudimentary programming language called Script, second generation
blockchains like Ethereum \cite{eth-whitepaper} feature
Turing-complete programming facilities, called 
\textit{smart contracts}. They provide robust trustworthy distributed computing
facilities even though the programs run on a peer-to-peer network with
untrusted participants. Each peer in the network runs the same program
and uses cryptographic methods to check the results among the other
peers and to create a persistent ledger of all transactions, the
blockchain, thus ensuring the integrity of the results. Third
generation blockchains, like Tezos \cite{tezos-whitepaper}, are adaptable to
new requirements without breaking participating peers (no  ``soft
forks'' required and ``hard forks'' can be avoided).

The strength of programs on the blockchain is also their
weakness. They are fully deterministic in that they can only depend on
data that is ultimately stored on the chain including the parameters
of a contract invocation. Moreover, the code, the data, as well as all
transactions are public. These properties make it hard to react to
external stimuli like time triggers or events like a price exceeding a
threshold unless these stimuli get translated to
contract invocations. 

Arguably, smart contracts are more useful if they can be
integrated with traditional software systems and thus triggered from
outside the blockchain.  Oracles
\cite{oracle-patterns,call-action-oracle} provide an approach for
contracts to obtain outside information. A contract registers a
request and a callback with an oracle. The oracle invokes the callback
as soon as an answer is available. 

There are other usecases for connecting a contract with traditional software. One example is
automating procedures like managing an auction,  bidding in
an auction, optimizing fees, or initiating delivery of goods to a
customer. While some of these procedures are amenable to
implementation as contracts, we might want to save the fee of running
them on the blockchain. In particular, for actions that happen
strictly within a single domain of trust, it is not worth running them
on the blockchain. For example, automated bidding runs on behalf of
a single peer.

Building such automation requires a programmatic interface to implement the interactions. 
Current blockchains mostly provide RPC interfaces, 
such as the Ethereum JSON-RPC API \cite{ethereum-rpc} and the Tezos RPC API \cite{tezos-whitepaper}, but they
require cumbersome manipulation of string data in JSON format and do
not provide static guarantees (except that the response to a
well-formed JSON input is also a well-formed JSON output). 
To improve on this situation
we present a typed API for invoking contracts from OCaml programs. 
Our typed API supports the implementation of application programs and oracles 
that safely interact with smart contracts on the blockchain. 
Moreover, our approach provides a type-safe facility to communicate
with contracts where data is automatically marshalled between OCaml
and the blockchain. This interface is a step towards a seamless
integration of contracts into traditional programs.

\subsubsection{Contributions}
\label{sec:contributions}

\begin{itemize}
\item A typed API for originating and invoking contracts as well as
  querying the state of the blockchain.
\item An operational semantics for functional programs running
  alongside smart contracts in a blockchain. 
\item Established various properties of the combined system with
  proofs in upcoming techreport.
\item An implementation of a low-level OCaml-API to the Tezos
  blockchain, which corresponds to the operational semantics.\footnote{%
    Available at \url{https://github.com/tezos-project/Tezos-Ocaml-API}.}
\end{itemize}
There is an extended version of the paper with further
proofs.\footnote{Available at \url{https://arxiv.org/abs/2108.11867}.}

\section{Motivation}
\label{sec:motivation}

Suppose you want to implement a bidding strategy for an auction that
is deployed on the blockchain as a smart contract.
Your bidding strategy may start at a certain amount and increase the bid until a
limit is reached. Of course, you only want to increase your bid if
someone else placed a higher bid. So you want to write a
program to implement this strategy.

This task cannot be implemented as a smart contract without
cooperation of the auction contract because it reacts
on external triggers.
Bidding requires watching the current highest bid of the contract and
react if another bidder places a higher bid.
The auction contract could anticipate the need for such observations
by allowing bidders to register callbacks that are invoked when a
higher bid arrives. However, we cannot assume such cooperation of the
auction contract nor would we be willing to pay the the fee for running
that callback.

\begin{lstlisting}[language=michelson,numbers=none,float={tp},caption={Header of the auction contract},label={lst:auction-contract-header},captionpos=b,emph={close,bid},emphstyle=\underbar]
parameter (or (unit %close)
              (unit %bid)); # bid in transfer
storage (pair bool          # bidding allowed
         (pair address      # contract owner
          address           # highest bidder's address
        ));
\end{lstlisting}
\begin{lstlisting}[language=Caml,float={tp},caption={Getting the auction handle},label={lst:getting-auction-handle},numbers=none]
# let auction = Cl.make_contract_hash auction_hash
#     ~parameter:(Ct.Or (Ct.Unit, Ct.Unit))
#     ~storage:(Ct.Pair (Ct.Bool, Ct.Pair (Ct.Addr, Ct.Addr))));;
val auction :
  ((unit, unit) Either.t,
   bool * (Cl.Addr.t * Cl.Addr.t)) Cl.contract
\end{lstlisting}
\begin{lstlisting}[language=Caml,float={tp},caption={Bidding strategy},label={lst:bidding-strategy},numbers=none]
let rec poll limit step =
  let (bidding, (_, highest_bidder)) = Cl.get_storage auction;
  let high_bid = Cl.get_balance auction;
  if bidding && high_bid < limit then
    (if highest_bidder <> my_address then  (* entrypoint %bid *)
      try
        Cl.call_contract auction
          (right (min (high_bid + step, limit))) 
      with
      | Cl.FAILWITH message -> poll limit step;
    Time.sleep(5 * 60);
    poll limit step)
\end{lstlisting}
For concreteness, Listing~\ref{lst:auction-contract-header} shows the
header of an auction contract in Michelson \cite{michelson}. The
\lstinline/parameter/ clause specifies the contract's parameter
type. It is a sum type (indicated by \lstinline/or/) and each
alternative constitutes an entrypoint, named \lstinline/%close/ and
\lstinline/%bid/. The caller selects the entrypoint by injecting the argument into the
\lstinline/left/ or \lstinline/right/ summand. Both entrypoints take a \lstinline/unit/
parameter. The \lstinline/%bid/ entrypoint considers the transferred
tokens as the bid.
The \lstinline/storage/ clause declares the state of the contract,
which is a nested pair type indicating whether bidding is allowed
(\lstinline/bool/), the address of the contract owner (to prohibit
unauthorized calls to \lstinline/%close/), and the bidder's
address. The highest bid corresponds to the token balance of the contract.

We only outline the implementation of the entrypoints. The
\lstinline/%close/ entrypoint first checks its sender's address
against the owner's address in the store. Then it transfers the funds
to the owner, closes the contract by clearing the bidding flag, and
leaves it to the owner to deliver the
goods.\footnote{For simplicity we elide safeguarding by a third-party oracle.}
The \lstinline/%bid/ entrypoint immediately returns each bid that is not higher
than the existing highest bid. Otherwise, it keeps the funds
transferred, returns the previous highest bid to its owner, and stores
the current bidder as the new highest bidder.

We present a program that implements
strategic bidding by interacting with the
blockchain.  The bidding strategy cannot be implemented as a smart
contract.

In Listing~\ref{lst:getting-auction-handle}, we use the library function
\lstinline!Cl.make_contract_from_hash! to 
obtain a typed handle for the contract.\footnote{\lstinline!Cl! is the module containing the contract
  library.} The function
takes the hash of the contract along with representations of
the types of the parameter and the storage (from module
\lstinline!Ct!). It checks the validity of the hash and the types with
the blockchain and returns a typed
handle, which is indexed with OCaml types corresponding to parameter
and storage type.




The implementation of the bidding strategy in Listing~\ref{lst:bidding-strategy} first checks the state of the
contract to find the current highest bid. As long as bidding is
allowed and the current bid is below our \lstinline/limit/, we update
our bid by a given amount \lstinline/step/, and then keep watching the state of the contract by
polling it every five minutes. 

The functions \lstinline/get_storage/ and \lstinline/get_balance/
obtains the storage and current balance, respectively, of a contract from
the blockchain. They never fail. Function \lstinline/call_contract/
takes a typed handle and a parameter of suitable type. It indicates
failure by raising an exception. If failure is caused by the
\lstinline/FAILWITH/ instruction in the contract, then the
corresponding \lstinline/Cl.FAILWITH/ exception is raised, which
carries a string corresponding to the argument of the instruction. In
our particular example, the auction may fail with signaling the message
\lstinline/"closed"/ or \lstinline/"bid too low"/. Our
code ignores this message for simplicity.

This code is idealized in several respects. Originating or running a contract requires proposing a fee to the blockchain, which may or may not be
accepted.
Starting a contract may also time out for a variety of reasons. So just
invoking a contract with a fixed fee does not guarantee the contract's
execution.  Even if the invocation is locally accepted, it still takes
a couple of cycles before we can be sure the invocation is globally
accepted and incorporated in the blockchain.
Hence, after starting the invocation, we have to observe the fate of
this invocation. If it does not get incorporated, then we need to
analyze the reason and react accordingly. For example, if the
invocation was rejected because of an insufficient fee, we might want
to restart with an increased fee. Or we might decide to wait until the
invocation goes through without increasing the fee.

Hence, we would implement a scheme similar to
the bidding strategy: start with a low fee and increase (or wait) until the
contract is accepted or a fee cap is reached. On the other hand, an
observer function like \lstinline/get_state/ always succeeds.

The low-level interface that we propose in this paper requires the
programmer to be explicit about fees, waiting, and polling the state
of contract invocations. 

In summary, a useful smart-contract-API  has facilities to
\begin{itemize}
\item query the current state of the blockchain (e.g., fees in the
  current block),
\item query storage and balance of a contract (to obtain
  the current highest bid),
\item originate contracts, invoke contracts, and initiate
  transfers. Hence, the API has to run on
  behalf of some account (by holding its private key).
\end{itemize}
These facilities are supported by the (untyped) RPC interface of the Tezos
blockchain, which is the basis of our implementation. 

\section{Execution Model}
\label{sec:execution-model}
The context of our work is the Tezos blockchain
\cite{tezos-whitepaper,tezos-intropaper}. Tezos is a self-amending
blockchain that improves several aspects compared to established
blockchains. Tezos proposes an original consensus algorithm,
Liquid Proof of Stake, that applies not only to the state of
its ledger, like Bitcoin \cite{bitcoin-whitepaper} or Ethereum
\cite{eth-whitepaper}, but also to upgrades of the protocol and the software.

Tezos supports two types of accounts: implicit accounts, which
are associated with a pair of private/public keys, and smart
contracts, which are programmable accounts created by an origination
operation. The address of a smart
contract is a unique public hash that depends on the creation
operation. No key pair is associated with a smart contract. An
implicit account is maintained on the blockchain with 
its public key and balance.
A smart contract account is stored with its script, storage,
and balance. A contract script maps a pair of a parameter and a
storage, which have fixed and monomorphic types, to a pair of a list of
internal operations and an updated storage. An account can
perform three kinds of transactions: (1) transfer tokens to an
implicit account, (2) invoke a smart contract, or (3) originate a new
smart contract. A contract origination specifies the script of the
contract and the initial contents of the contract storage, while a
contract invocation must provide input data. Each transaction contains
a fee to be paid either by payment to a baker or by destruction
(burning). A transaction is injected into the blockchain network via a
node, which then validates the transaction before submitting it to the
network. A transaction may be rejected by the node for a
number of reasons. After validation, the transaction is injected into
a \emph{mempool}, which contains all pending transactions before they can
be included in a block. A pending transaction may simply disappear
from the mempool, for example, a transaction times out when 60 blocks
have passed and it can no longer be included in a block. When a transaction
is included in the blockchain, the affected accounts are updated
according to the transaction result.

The execution model consists of functional (OCaml) programs that
interact with an abstraction of the Tezos blockchain
\cite{tezos-whitepaper}. As the blockchain is realized by a
peer-to-peer network of 
independent nodes, interaction happens through
\emph{local nodes} that receive requests to originate and invoke
contracts from programs that run on a particular node. We model the blockchain
itself as a separate, abstract global entity that represents the
current consensual state of the system. Our model does not express
low-level details, but relies on nondeterminism to describe
the possible behaviors of the system. In particular, we do
\textbf{not} formalize the execution of the smart contracts
themselves, we rather consider them as black boxes and probe their
observable behavior. Tezos's smart contract language Michelson and its
properties have been formalized elsewhere \cite{DBLP:conf/fm/BernardoCHPT19}.

We write $\emptyset$ for the empty set and $\mathbf{e :: s}$ to
decompose a set nondeterministically into an element $\mathbf{e}$ and a set
$\mathbf{s}$. We generally use lowercase boldface for metavariables
ranging over values of a certain syntactic category, e.g., \PUK\ for
public keys, and the capitalized name for the
corresponding type as well as for the set of these values (as in \TPUK).

\subsection{Local Node}
\label{sec:local-node}

A local node runs on behalf of authorities, which are called
\emph{accounts} in Tezos. An account is represented by a key pair $\Angle{
  \PAK, \PUK
} $, where $\PAK$ is a private key and $\PUK$ the corresponding public
key in a public key encryption scheme.

The local node offers operations to transfer tokens from one account
to another, to invoke a contract, and to originate a contract on the blockchain.
\begin{align*}
  \OP &::= \TRANSFER[\PARAMETER]\NTEZ\PUK\ADDR\MTEZ
  \\&\mid \ORIGINATE\NTEZ\PUK\CODE\INIT\MTEZ
\end{align*}
In the transfer, which also serves as contract invocation, \NTEZ\ is the amount of tokens transferred, \PUK\ is the public key of the
sender, \ADDR\ is either a public 
key for an implicit account (in case of a simple transfer)  or a
public hash for a smart contract (for an invocation), \PARAMETER\ is
the argument passed to the smart contract, 
which is empty for a simple transfer, and \MTEZ\ is
the amount of tokens for the transaction fee. In originate,  \CODE\ is
the script of a smart contract and \INIT\ is the initial value of the
contract's storage. Each operation returns an \emph{operation hash}
\OPH, on which we can query the status of the operation.

The local node offers several ways to query the current state of
the blockchain.
Some \emph{query operators} are defined by the following grammar:
\begin{align*}
  \QOP & ::=
         \GETBALANCE{}
  \mid \GETSTATUS{}
  \mid \GETSTORAGE{}
         \mid \GETCONTRACT{}
         \mid \dots
\end{align*}
We obtain the balance associated with an implicit account or a
contract by its public key or public hash, respectively; the status of
a submitted operation by its operation hash; the stored value of a
contract by its public hash; and the public hash of a contract by the
operation hash of its originating transaction.

The domain-specific types come with different guarantees. Values of
type \TPUH\ and \TPUK\ as well as \TADDR\ are not necessarily valid,
as there might be no contract associated with a hash / no account
associated with a public key. In contrast, a value of type
$\TCONTRACT\ \TYPE$ is a public hash that is verified to be associated
with a contract with parameter type $\TYPE$. Operation hashes \OPH\
are only returned from blockchain operations. As the surface language
neither contains literals of type \TOPH\ nor are there casts into
that type, all values of \TOPH\ are valid. 

\begin{definition}
  The \emph{state of a node} is a pair
  $\NODE = [ \EXPRS, \ACCOUNTS
  ]$, where $\EXPRS$ is a
  set of programs and $\ACCOUNTS \subseteq \TPAK \times \TPUK$  is a set of
  implicit accounts.
\end{definition}

\begin{figure}[tp]
\begin{align*}
  \CONSTANT & ::= \INT \mid \FIX
              \mid \OPH \mid \PUH \mid\PUK \mid \CODE \mid
              \NTEZ \mid \SUNIT \mid \FALSE \mid \TRUE \\
  \STATUS & ::= \STATUSPENDING \mid \STATUSINCLUDING (\INT) \mid
            \STATUSTIMEOUT \\
  \ERROR &::= \ERRPRG \mid \ERRBAL \mid \ERRCOUNT \mid \ERRFEE \mid
           \ERRPUK \mid \ERRPUH \mid \ERRARG \mid \ERRINIT\\
	\EXPR & :: =  \CONSTANT \mid \STATUS \mid \ERROR \mid \VARIABLE \mid \lambda \VARIABLE. \EXPR
	\mid \EXPR\EXPR  
  \mid \EXPR\ \PLUS\ \EXPR\ \mid
  \EXPR\ \EQUAL\ \EXPR\
  \mid \EXPR\ \AND\ \EXPR\ \mid \EXPR\ \OR\ \EXPR\ \mid \NOT\ \EXPR
  \\& \mid (\EXPR,\EXPR)\mid \NIL\ \mid \CONS\ \EXPR\ \EXPR \mid \LEFT\ \EXPR \mid\RIGHT\
  \EXPR \mid \SOME\ \EXPR \mid \NONE
  \mid \MATCH\ \EXPR\ \WITH\ \PATTERN\to\EXPR \dots 
  \\& \mid \RAISE\ \EXPR \mid \TRY\ \EXPR\ \EXCEPT\ \EXPR \mid \CAST\EXPR\TYPE\TYPEU
  \\& \mid \QOP\  \EXPR \mid \TRANSFER[\EXPR]\EXPR\EXPR\EXPR\EXPR
  \\& \mid \ORIGINATE\EXPR\EXPR\EXPR\EXPR\EXPR
  \\
  \PATTERN &::= \VARIABLE \mid (\PATTERN, \PATTERN) \mid \NIL \mid \CONS\ \PATTERN\ \PATTERN \mid \LEFT\ \PATTERN \mid\RIGHT\
             \PATTERN \mid \SOME\ \PATTERN \mid \NONE \\
            &\mid \FALSE \mid \TRUE \mid \STATUS \mid \ERROR
  \\[2ex]
  \TYPE, \TYPEU & ::=
                  \TPUH \mid
                  \TPUK \mid
                  \TADDR \mid 
                  \TCONTRACT\ \TYPE\ \TYPEU \mid
                  \TCODE\ \TYPE\ \TYPEU \mid
                  \TOPH\ \TYPE\ \TYPEU \mid
                  \TSTATUS \mid \TEXCEPTION \mid \TTEZ \\
  & \mid \TNO \mid \TINT \mid \TUNIT \mid \TBOOL \mid \TSTRING \mid \TYPE\to\TYPEU \mid \TPAIR\ \TYPE\ \TYPEU \mid \TLIST\ \TYPE
    \mid \TSUM\ \TYPE\ \TYPEU \mid \TOPTION\ \TYPE 
\end{align*}
  \caption{Syntax of expressions, \EXPR, and types, \TYPE}
  \label{fig:syntax-expressions}
\end{figure}
Queries and operations are started by closed
expressions of type unit that run on the
local node. Each program can send transactions on behalf of any
account on the local node.  Figure~\ref{fig:syntax-expressions}
defines the syntax of lambda calculus with sum, product, list, and option types, exceptions and fixpoint.
Pattern matching is the only means to decompose values, cf.\ \PATTERN. The execution model envisions off-chain programs interacting with smart contracts on the blockchain. The programs are defined using expression scripts. The off-chain scripts run on behalf of a single entity.

Domain-specific primitive types and constants \CONSTANT\ support blockchain
interaction, as well as several exceptional values collected in \ERROR.
There is syntax to initiate transfers and to originate contracts as
well as for the queries. Finally, there is a type cast
$\CAST\EXPR\TYPE\TYPEU$, which we describe after discussing types. An
implementation provides all of these types and operations via a library API.

Types (also in Figure~\ref{fig:syntax-expressions}) comprise some base types as well as 
functions, pairs, lists, sums, and option types. These types are
chosen to match with built-in types of Michelson. 
There are domain specific types of public hashes $\TPUH$ and public keys $\TPUK$ subsumed by
a type of addresses $\TADDR$. $\TCONTRACT\ \TYPE\ \TYPEU$ is the type of
a contract with parameter type $\TYPE$ and storage type $\TYPEU$. $\TCODE\ \TYPE\ \TYPEU$
indicates a Michelson program with parameter type $\TYPE$ and storage
type $\TYPEU$. Tezos tokens have type $\TTEZ$. The type $\TOPH\ \TYPE\
\TYPEU$ signifies operation hashes
returned by blockchain operations. The parameters of the hash carry
the types when originating a contract. Otherwise, they are set to the
irrelevant type $\TNO$. We take the liberty of omitting irrelevant
type parameters, that is, we write $\TOPH$ for $\TOPH\ \TNO\ \TNO$.
Querying the status of an operation
returns a value of type $\TSTATUS$. Exceptions have type
$\TEXCEPTION$. 

\begin{figure}[tp]
\begin{align*}
  \ECN{} & :: = \ECHOLE \mid \SC[\overline\VAL\ \ECN{}\ \overline\EXPR] \mid \RAISE\ \ECN{} \mid \TRY\ \ECN{}\ \EXCEPT\ \EXPR
  \mid \MATCH\ \ECN{}\ \WITH\ \PATTERN\to\EXPR \dots
  \\
  \VAL & ::= \CONSTANT \mid \STATUS \mid \ERROR \mid \lambda x.\EXPR \mid(\VAL, \VAL) \mid
         \NIL \mid \CONS\ \VAL\ \VAL \mid \LEFT\ \VAL \mid \RIGHT\
         \VAL \mid \SOME\ \VAL\mid \NONE
\end{align*}
  \caption{Evaluation contexts and values}
  \label{fig:evaluation-contexts-values}
\end{figure}
Figure~\ref{fig:evaluation-contexts-values} defines evaluation
contexts $\ECS$ and values $\VAL$.
Here $\SC$ ranges over the remaining syntactic constructors, which are
treated uniformly: evaluation proceeds from left to right. Values are
standard for call-by-value lambda calculus.


Type casts are only applicable to certain pairs of
types governed by a relation $\SubType$, which could also serve as a
subtyping relation. It is given by the axioms
$\inferrule{}{\TPUH \SubType \TADDR}$, $\inferrule{}{\TPUK \SubType
  \TADDR}$, and $\inferrule{}{\TCONTRACT\ \TYPE\ \TYPEU \SubType \TPUH}$.
A cast from $\TYPE$ to $\TYPEU$ is only allowed
if $\TYPE \SubType \TYPEU$ (upcast) or $\TYPEU \SubType \TYPE$
(downcast). Upcasts always succeed, but downcasts may fail at run time.
In particular,  public hashes and public keys can both stand for addresses. Moreover, a smart
contract with parameter type $\TYPE$ is represented by its public hash
at run time. The corresponding downcast must check whether the public
hash is valid and has the expected parameter and storage type.

Figure~\ref{fig:typing-expressions} presents selected typing rules for
expressions. We rely on an external typing judgment
$\JTypeCode\CODE\TYPE$ for the contract language, which we leave
unspecified, and $\JTypeValue\STRING\TYPE$ for serialized values as
stored on the blockchain. The latter judgment states $\STRING$ is
a string parseable as a value of type $\TYPE$.
\begin{figure}[tp]
  \begin{mathpar}
    \inferrule{}{\JTypeExpr\TEnv\INT\TINT}

    \inferrule{}{\JTypeExpr\TEnv\OPH{\TOPH\ \TYPE\ \TYPEU}}

    \inferrule{}{\JTypeExpr\TEnv\PUH{\TPUH}}

    \inferrule{}{\JTypeExpr\TEnv\PUK\TPUK}

    \inferrule{
      \JTypeCode \CODE{ \TPAIR\ \TYPE_p\ \TYPE_s}
    }{\JTypeExpr\TEnv\CODE{\TCODE\ \TYPE_p\ \TYPE_s}}

    %
    \inferrule{}{\JTypeExpr\TEnv\NTEZ\TTEZ}

    \inferrule{}{\JTypeExpr\TEnv\SUNIT\TUNIT}

    \inferrule{}{\JTypeExpr\TEnv\FALSE\TBOOL}

    \inferrule{}{\JTypeExpr\TEnv\TRUE\TBOOL}

    \inferrule{}{\JTypeExpr\TEnv\STATUSPENDING\TSTATUS}

    \inferrule{}{\JTypeExpr\TEnv\STATUSTIMEOUT\TSTATUS}

    \inferrule{\JTypeExpr\TEnv\EXPR\TINT}{\JTypeExpr\TEnv{\STATUSINCLUDING
        (\EXPR)}\TSTATUS}

    \inferrule{}{\JTypeExpr\TEnv\ERROR\TEXCEPTION}

    \inferrule{}{\JTypeExpr\TEnv\VARIABLE{\TEnv (\VARIABLE)}}

    \inferrule{ \JTypeExpr{\TEnv, \VARIABLE:\TYPE'}\EXPR{\TYPE} }{
      \JTypeExpr\TEnv{\lambda\VARIABLE.\EXPR}{\TYPE'\to\TYPE}}

    \inferrule{
      \JTypeExpr\TEnv\EXPR{\TYPE'\to\TYPE} \\
      \JTypeExpr\TEnv{\EXPR'}{\TYPE'} }{ \JTypeExpr\TEnv{\EXPR\
        \EXPR'}\TYPE }

    \inferrule{
      \JTypeExpr\TEnv\EXPR\TYPE \\
      \JTypeExpr\TEnv{\EXPR'}{\TYPE'} }{ \JTypeExpr\TEnv{(\EXPR,
        \EXPR')}{\TPAIR\ \TYPE\ \TYPE'} }
  %
  %

    \inferrule{ \JTypeExpr\TEnv\EXPR\TEXCEPTION }{
      \JTypeExpr\TEnv{\RAISE\ \EXPR}\TYPE }

    \inferrule{
      \JTypeExpr\TEnv\EXPR\TYPE \\
      \JTypeExpr\TEnv{\EXPR'}{\TEXCEPTION\to\TYPE} }{
      \JTypeExpr\TEnv{\TRY\ \EXPR\ \EXCEPT\ \EXPR'}\TYPE }

    \inferrule{
      \JTypeExpr\TEnv\EXPR\TYPE \\
      \TYPE \SubType \TYPEU \vee \TYPEU \SubType \TYPE
    }{
      \JTypeExpr\TEnv{\CAST\EXPR\TYPE\TYPEU}\TYPEU
    }
  \end{mathpar}
  \caption{Typing rules for expressions (excerpt)}
  \label{fig:typing-expressions}
\end{figure}
\begin{figure}[tp]
  \begin{mathpar}
    \inferrule{
      \JTypeExpr\TEnv{\EXPR_1}\TTEZ \\
      \JTypeExpr\TEnv{\EXPR_2}\TPUK \\
      \JTypeExpr\TEnv{\EXPR_3}\TPUK \\
      \JTypeExpr\TEnv{\EXPR_4}\TUNIT \\
      \JTypeExpr\TEnv{\EXPR_5}\TTEZ }{
      \JTypeExpr\TEnv{\TRANSFER[\EXPR_4]{\EXPR_1}{\EXPR_2}{\EXPR_3}{\EXPR_5}}\TOPH\
      \TNO\ \TNO
    }

    \inferrule{
      \JTypeExpr\TEnv{\EXPR_1}\TTEZ \\
      \JTypeExpr\TEnv{\EXPR_2}\TPUK \\
      \JTypeExpr\TEnv{\EXPR_3}\TCONTRACT\ \TYPE_p\ \TYPE_s \\
      \JTypeExpr\TEnv{\EXPR_4}{\TYPE_p} \\
      \JTypeExpr\TEnv{\EXPR_5}\TTEZ }{
      \JTypeExpr\TEnv{\TRANSFER[\EXPR_4]{\EXPR_1}{\EXPR_2}{\EXPR_3}{\EXPR_5}}{\TOPH\
      \TNO\ \TNO}
    }

    \inferrule{
      \JTypeExpr\TEnv{\EXPR_1}\TTEZ \\
      \JTypeExpr\TEnv{\EXPR_2}\TPUK \\
      \JTypeExpr\TEnv{\EXPR_3}{\TCODE\ \TYPE_p\ \TYPE_s} \\
      \JTypeExpr\TEnv{\EXPR_4}\TYPE_s \\ 
      \JTypeExpr\TEnv{\EXPR_5}\TTEZ }{
      \JTypeExpr\TEnv{\ORIGINATE{\EXPR_1}{\EXPR_2}{\EXPR_3}{\EXPR_4}{\EXPR_5}}{\TOPH\
      \TYPE_p\ \TYPE_s}
    }
  \end{mathpar}
\begin{mathpar}
  %
  \inferrule{
    \JTypeExpr\TEnv\EXPR\TADDR
  }{
    \JTypeExpr\TEnv {\GETBALANCE\EXPR}\TTEZ
  }

  \inferrule{
    \JTypeExpr\TEnv\EXPR{\TOPH\ \TYPE\ \TYPEU}
  }{
    \JTypeExpr\TEnv {\GETSTATUS\EXPR}\TSTATUS
  }

  \inferrule{
    \JTypeExpr\TEnv\EXPR{\TCONTRACT\ \TYPE_p\ \TYPE_s}
  }{
    \JTypeExpr\TEnv {\GETSTORAGE\EXPR}\TYPE_s
  }

  %
  %
  \inferrule{
    \JTypeExpr\TEnv\EXPR{\TOPH\ \TYPE\ \TYPEU} \\
    \TYPE\ne\TNO \\ \TYPEU \ne\TNO
  }{
    \JTypeExpr\TEnv {\GETCONTRACT\EXPR}{\TCONTRACT\ \TYPE\ \TYPEU}
  }
\end{mathpar}
  \caption{Typing rules for blockchain operations and queries}
  \label{fig:typing-blockchain-operations}
\end{figure}



\subsection{Global Structures}
\label{sec:global}

Our execution model abstracts from the particulars of the blockchain
implementation, like the peer-to-peer structure or the distributed
consensus protocol. Hence, we represent the blockchain by a few global
entities: managers, contractors, and a pool of operations. 

A \emph{manager} keeps track of a single implicit account. Managers are
represented by a partial map $\MANAGERS : \TPUK \partialto \TBAL
\times \TCOU$. If $\MANAGERS (\PUK) = \Angle{\BAL, \COU}$ is defined, then  $\PUK$ is the
public key of an account, $\BAL$ is its
balance and $\COU$ is its counter whose form is a value-flag pair
$(n,b) \in \Nat\times\TBOOL$, where $n$ is the value of the counter
and ${b}$ is its flag.   The counter is used internally to serialize transactions.

A \emph{contractor} manages a single smart contract. Contractors are
represented by a partial map $\CONTRACTORS : \TPUH \partialto
\TCODE \times \TIME  \times \TBAL \times  \TSTORAGE$. If $\CONTRACTORS (\PUH)$ = $\langle
\CODE,$ $\TIME, \BAL, \STORAGE \rangle $ is defined, then $\PUH$ is the
public hash of a contract,
$\CODE$ is its  code,
$\TIME$ is the time when it was accepted,
$\BAL$ is its current balance,
and $\STORAGE$ is its current storage. The hash $\PUH$ is
self-verifying as it is calculated from the fixed components $\CODE$
and $\TIME$. All time stamps will be different in our model.

When an operation is started on a node, it enters a \emph{pool} as a
pending operation. A pending operation is either dismissed after some time or
promoted to an included operation, which has become a permanent part
of the blockchain.

The pool is a partial map $  \PENDING= \TOPH \partialto
\TOP\times \TTIME \times \TSTATUS$ where
\begin{align*}
  \TSTATUS&= \STATUSPENDING + \STATUSINCLUDED\ \TTIME + \STATUSTIMEOUT
\end{align*}
such that
if $\PENDING(\OPH) = \langle  \OP, \TIME, \STATUS
\rangle $ is defined, then $\OPH$ is the public hash of the operation, $\OP$ is the operation, $\TIME$ is the time when the
operation was injected, and $\STATUS$ is either $\STATUSPENDING$,
$\STATUSINCLUDED\ \TIME'$, or $\STATUSTIMEOUT$.
A pool $\PENDING$ is \emph{well-formed} if, for all $\OPH$, $\PENDING (\OPH) =
\Angle{\OP,\TIME, \STATUSINCLUDED\ \TIME'}$ implies $\TIME' \ge
\TIME$ and $\OPH = \GENERATEOPH (\OP, \TIME)$.

A \emph{pending operation} is represented by
$\OPH \mapsto \langle  \OP, \TIME, \STATUSPENDING\rangle $.
Once the operation is accepted, it changes its status to included:
$\OPH \mapsto \langle  \OP, \TIME, \STATUSINCLUDED\ \TIME'\rangle $, where
$\TIME' \ge \TIME$ is  when the operation was included in the
blockchain.  The operation may also be dropped at any time, which is
represented by
$\OPH \mapsto \langle  \OP, \TIME, \STATUSTIMEOUT\rangle $. There are several
causes for dropping, primarily timeout or overflow of the pending
pool which is limited in size in the implementation.



In summary, the \emph{state of a blockchain} is a tuple
$\BLOCKCHAIN = [\PENDING, \MANAGERS, \CONTRACTORS, \TIME]$ where
$\PENDING$ is a pool of operations, $\MANAGERS$ is a map of managers,
$\CONTRACTORS$ is a map of contractors, and $\TIME$ is the current
time. 

We often use the dot notation to project a component from a tuple. For
instance, we write $\BLOCKCHAIN.\MANAGERS$ to access the managers
component. 

A \emph{blockchain configuration} has the form
$\BLOCKCHAIN[ \NODE_1, \dots, \NODE_n]$, for some $n>0$, where $\BLOCKCHAIN$ is a
blockchain and the $\NODE_i$ are local nodes, for $1\le i\le n$.
In a \emph{well-formed configuration}, the accounts on the local nodes are all different and each
local account has a manager in $\BLOCKCHAIN$:
\begin{enumerate}
\item for all $1\le i< j\le n$, $\NODE_i.\ACCOUNTS \cap
  \NODE_j.\ACCOUNTS = \emptyset$;
\item for all $1\le i \le n$, 
    $\forall a \in \NODE_i.\ACCOUNTS \implies a.\PUK \in \DOM( \BLOCKCHAIN.\MANAGERS)$.
\end{enumerate}

\section{Operational Semantics}
\label{sec:transitions}

The operational semantics is defined by several kinds of transitions:
\begin{enumerate}
\item $\ExprTrans$ single-step evaluation of an expression in a local node,
\item $\NodeTrans$ internal transitions of a node,
\item $\BlockTrans$ transitions of the blockchain state,
\item $\SystemTrans$ blockchain system transitions.
\end{enumerate}

Evaluation of expressions is standard for call-by-value lambda calculus defined
using evaluation contexts
$\EC{}$. Figure~\ref{fig:expression-reduction} shows some of the
reduction rules. 
\begin{figure}[tp]
  \begin{mathpar}
    \inferrule{}{ \EC{(\lambda x.\EXPR)\VAL} 
      \ExprTrans 
      \EC{\EXPR[\VAL/x]} }

    \inferrule{}{ \EC{\TRY\ \VAL\ \EXCEPT\ \EXPR} \ExprTrans \EC{\VAL} }

  %
  %
    \inferrule{\TYPE \SubType \TYPEU}{ \EC{\CAST\VAL\TYPE\TYPEU} \ExprTrans \EC{\VAL}}

    \inferrule{ \TRY \notin \EC[F]{} }{ \EC{\TRY\ {\EC[F]{\RAISE\
            \VAL}}\ \EXCEPT\ \EXPR} \ExprTrans \EC{\EXPR\ \VAL} }
    %
  \end{mathpar}
  \caption{Select expression reduction rules (pattern matching omitted)}
  \label{fig:expression-reduction}
\end{figure}
The internal transitions of a node are just evaluation of expressions.
\begin{mathpar}
  \inferrule[Node-Eval]
  {
    \EXPR \ExprTrans \EXPR'
  }{
    [\EC\EXPR :: \EXPRS, \ACCOUNTS] \NodeTrans{}
    [\EC{\EXPR'} :: \EXPRS, \ACCOUNTS]
  }
\end{mathpar}
The remaining transitions affect a local node in the context of the
blockchain. To this end, any local node may be selected.
\begin{mathpar}
  \inferrule[Config-System]{
    \NODE\|\BLOCKCHAIN \SystemTrans \NODE'\|\BLOCKCHAIN'
}{
    \BLOCKCHAIN[\NODE :: \overline{\NODE}] \SystemTrans
    \BLOCKCHAIN'[\NODE' :: \overline{\NODE}]
  }

  \inferrule[Config-Node]
  {\NODE \NodeTrans \NODE'}
  { {\BLOCKCHAIN[\NODE :: \overline\NODE]}
    \SystemTrans
    {\BLOCKCHAIN[\NODE' :: \overline\NODE]}}

  \inferrule[Config-Block]
  {\BLOCKCHAIN \BlockTrans \BLOCKCHAIN'}
  { \BLOCKCHAIN[{\overline\NODE}]
    \SystemTrans
    \BLOCKCHAIN'[{\overline\NODE}]}
\end{mathpar}

\begin{figure}[tp]
\begin{mathpar}
  \inferrule[Node-Inject]{
    \Angle{\PAK,\PUK} \in \ACCOUNTS \\
    \CHECKBAL (\MANAGERS, \PUK, \NTEZ, \MTEZ) \\
    \CHECKARG (\CONTRACTORS, \PUH, \PARAMETER) \\
    \CHECKCOU (\MANAGERS, \PUK) \\
    \CHECKPUH (\CONTRACTORS, \PUH) \\
    \CHECKGAS (\CONTRACTORS, \PUH, \PARAMETER, \MTEZ) \\
    \OPH = \GENERATEOPH (\OP, \TIME) \\
    \OP = \TRANSFER[\PARAMETER]\NTEZ\PUK{\PUH}\MTEZ    
  }{
    { [\EC\OP :: \EXPRS, \ACCOUNTS] \|
      [\PENDING, \MANAGERS, \CONTRACTORS, \TIME] } \SystemTrans 
   { [\EC{\OPH}  :: \EXPRS, \ACCOUNTS] \|
     [ \OPH \mapsto \Angle{\OP, \TIME, \STATUSPENDING}
     ::\PENDING,} \\ {
     \UPDATECOU(\MANAGERS, \PUK, \TRUE),
     \CONTRACTORS,
     \TIME]
   }
 }

  \inferrule[Node-Reject]{
    \NEG\ \CHECKBAL (\BLOCKCHAIN.\MANAGERS, \OP.\PUK, \OP.\NTEZ, \OP.\MTEZ) 
    }{
    { [\EC\OP :: \EXPRS, \ACCOUNTS] \| 
      \BLOCKCHAIN
      } \SystemTrans
    { [\EC{\RAISE\ \ERRBAL} :: \EXPRS, \ACCOUNTS] \|
      \BLOCKCHAIN
    }
  }

  \inferrule[Block-Accept]{
    \OP = \TRANSFER[\PARAMETER]\NTEZ\PUK{\PUH}\MTEZ \\
    \TIME - \hat\TIME \le 60
  }{
    { 
      [\OPH \mapsto \Angle{\OP, \hat \TIME, \STATUSPENDING}
     ::\PENDING, \MANAGERS,
      \CONTRACTORS, \TIME]}
    \BlockTrans 
    {
      [\OPH \mapsto \Angle{\OP, \hat\TIME, \STATUSINCLUDING\ \TIME} :: \PENDING}, \\
      { \UPDATESUCC (\MANAGERS, \PUK, \NTEZ, \MTEZ), 
       \UPDATECONSTR (\CONTRACTORS, \PUH, \NTEZ, \PARAMETER), \TIME +1]
    }
  }

  \inferrule[Block-Timeout]{
    \TIME-\hat\TIME > 60
  }{ 
    {[\OPH \mapsto \Angle{\OP, \hat \TIME, \STATUSPENDING}
     ::\PENDING, \MANAGERS,
      \CONTRACTORS, \TIME]}
    \BlockTrans 
    { 
      [\OPH \mapsto \Angle{\OP, \hat \TIME, \STATUSTIMEOUT}
     :: \PENDING,} \\ {\UPDATECOU(\MANAGERS, \OP.\PUK, \FALSE),
      \CONTRACTORS, \TIME]}
  }
\end{mathpar}
  
  \caption{Lifecycle transitions of a transaction}
  \label{fig:lifecycle-transaction}
\end{figure}
Figure~\ref{fig:lifecycle-transaction} shows the transitions to start
and finalize a contract invocation.
\TirName{Node-Inject} affects a local node and the blockchain. It
nondeterminstically selects a program that wants to do a transfer
operation. It
checks whether the sender of the transfer is a valid local account, whether
the balance is sufficient to pay the fee and the transferred amount,
whether there is an active transition for this sender (chkCount),
whether the public hash is associated with a smart contract on the blockchain, whether the type of the input parameter matchs with the smart contract's parameter type (chkArg), and whether the fee is
sufficient. If these conditions are fulfilled, the transition forges
an operation hash and returns it to the local node. 
The pending operation enters the pool and the sender's counter is set
to indicate an ongoing transition.

We give just one example \TirName{Node-Reject} of the numerous
transitions that cover the cases where one of the 
premises of \TirName{Node-Inject} is not fulfilled. Each of them
raises an exception that describes which condition was violated.

Acceptance or rejection of a pending operation happens on the
blockchain independent of any local node. In our model, these
transitions are nondeterministic so that acceptance can happen any
time in the next 60 cycles \TirName{Block-Accept}. Afterwards, a
pending operation can only time out \TirName{Block-Timeout}.
If the transaction is accepted, then the sender's counter is reset,
the balances of sender is adjusted (updSucc), the smart contract's storage and balance are updated (updConstr), and the time stamp increases.

Whereas \TirName{Node-Inject} and \TirName{Block-Accept} are
particular to the transfer operation, the timeout transition applies
to all operations. It just changes the state of the operation and
resets the sender's counter, thus rolling back the transaction.


\subsection{Cast Reductions}
\label{sec:special-reductions}

\begin{figure}[tp]
\begin{mathpar}
  \inferrule[Contract-Yes]{
    \JTypeCode\CODE{\TPAIR\ \TYPE\ \TYPEU} \\
    \BLOCKCHAIN.\CONTRACTORS (\PUH) =  \Angle{\CODE, \tilde\TIME, \NTEZ', \STRING'}
    }{
    { [\EC{\CAST\PUH\TPUH{\TCONTRACT\ \TYPE}}  :: \EXPRS, \ACCOUNTS] \|
      \BLOCKCHAIN
    }
    \SystemTrans
    { [\EC{\PUH} :: \EXPRS, \ACCOUNTS] \|
      \BLOCKCHAIN
   }
  }

  %
  \inferrule[Contract-No]{
    \BLOCKCHAIN.\CONTRACTORS (\PUH) = \Angle{\CODE, \tilde\TIME,
      \NTEZ', \STRING'} \Rightarrow {}
        \JTypeCode\CODE{\TPAIR\ \TYPE'\ \TYPEU} \wedge \TYPE \ne \TYPE' \\
    }{
    { [\EC{\CAST\PUH\TPUH{\TCONTRACT\ \TYPE}}  :: \EXPRS, \ACCOUNTS] \|
      \BLOCKCHAIN
    }
    \SystemTrans
    { [\EC{\RAISE\ \ERRPRG} :: \EXPRS, \ACCOUNTS] \|
      \BLOCKCHAIN
   }
  }
\end{mathpar}
  \caption{Cast reductions (excerpt)}
  \label{fig:cast-reductions}
\end{figure}
\begin{figure}[t]
  \begin{mathpar}
    \inferrule[Block-Originate]{
      \Angle{\PAK,\PUK} \in \ACCOUNTS \\ \CHECKBAL (\MANAGERS, \PUK, \NTEZ, \MTEZ) \\
      \CHECKCOU (\MANAGERS, \PUK) \\
      \CHECKPRG (\CODE) \\
      \CHECKGAS (\CODE, \INIT, \NTEZ, \MTEZ)  \\
      \CHECKINIT (\CODE, \STRING) \\
      \OPH = \GENERATEOPH(\OP, \TIME) \\
      \OP = \ORIGINATE\NTEZ\PUK\CODE\STRING\MTEZ }{ {[\EC\OP :: \EXPRS,
      \ACCOUNTS
      ] \| [\PENDING, \MANAGERS, \CONTRACTORS, \TIME]}    
       \SystemTrans 
      {[\EC{\OPH} :: \EXPRS, \ACCOUNTS] \| [\OPH \mapsto \Angle{\OP,
        \TIME, \STATUSPENDING} ::\PENDING,} \\ {
      \UPDATECOU(\MANAGERS,\PUK,\TRUE), \CONTRACTORS, \TIME] }}

  \inferrule[Block-Originate-Accept]{
    \OP = \ORIGINATE\NTEZ\PUK\CODE\STRING\MTEZ \\
    \PUH = \GENERATEHASH(\CODE, \TIME) \\
    \TIME-\hat\TIME  \le 60
  }{
    {[\OPH \mapsto \Angle{\OP, \hat\TIME, \STATUSPENDING} :: \PENDING, \MANAGERS, \CONTRACTORS, \TIME]} \BlockTrans    
     {[\OPH \mapsto \Angle{\OP, \hat\TIME, \STATUSINCLUDING\ \TIME} :: \PENDING,} \\ { \UPDATESUCC
      (\MANAGERS, \PUK, \NTEZ, \MTEZ),} 
     {\PUH \mapsto  \Angle{\CODE, \TIME, \NTEZ, \STRING} :: \CONTRACTORS, \TIME+1]}
  }

    \inferrule[Block-Accept-Query]{
    \OP = \ORIGINATE\NTEZ\PUK\CODE\STRING\MTEZ \\
    \PENDING (\OPH) =  \Angle{\OP,
      \hat\TIME, \STATUSINCLUDING\ \tilde\TIME} \\
    \PUH = \GENERATEHASH (\CODE, \tilde\TIME)
}{
      [\EC{\GETCONTRACT\ \OPH}  :: \EXPRS, \ACCOUNTS] \| [\PENDING, \MANAGERS, \CONTRACTORS, \TIME]
    \SystemTrans
     [\EC{ \PUH}  :: \EXPRS, \ACCOUNTS] \|  [\PENDING, \MANAGERS, \CONTRACTORS, \TIME]
  }
\end{mathpar}
  \caption{Smart contract origination}
  \label{fig:contract-origination}
\end{figure}
\begin{figure}[t]
  \begin{mathpar}
    \inferrule[Query-Balance-Implicit]{
      \BLOCKCHAIN.\MANAGERS (\PUK) = \Angle{\BAL,\COU}
    }{[\EC{\GETBALANCE\PUK} :: \EXPRS, \ACCOUNTS] \| \BLOCKCHAIN
      \SystemTrans\ [\EC{\BAL} ::\EXPRS, \ACCOUNTS] \| \BLOCKCHAIN}

    \inferrule[Query-Balance-Fail]{ \PUK \notin \DOM
      (\BLOCKCHAIN.\MANAGERS) \ }{[\EC{\GETBALANCE\PUK} :: \EXPRS,
      \ACCOUNTS] \| \BLOCKCHAIN \SystemTrans {[\EC{\RAISE\ \ERRPUK}
        ::\EXPRS, \ACCOUNTS] \| \BLOCKCHAIN}}
  \end{mathpar}
  \caption{Example queries}
  \label{fig:example-queries}
\end{figure}

Figure~\ref{fig:cast-reductions} contains the most interesting
example of cast reductions, from a public hash to a typed
contract. These reductions force the local node to obtain information
from the blockchain. The cast succeeds on \PUH\ (`CONTRACT-YES'), if there is a
contractor for \PUH\ such that the stored code has the parameter type
expected by the cast. The cast fails (`CONTRACT-NO'), if \PUH\
is invalid or if the types do not match.

\subsection{Smart Contracts}
\label{sec:smart-contracts}

The invocation of smart contracts is similar to a transfer, so we
elide the details. Figure~\ref{fig:contract-origination} contains the
transition \TirName{Block-Originate} to originate a smart contract. The basic scheme is similar
to the transfer. The preconditions for the operation are checked, but
there are extra preconditions for origination:  the program must be
well-formed and typed, the initial storage value must match its
type. The operation ends up in the pool in pending status.

Acceptance of origination is slightly different as for transfers as
shown in \TirName{Block-Accept}. We calculate the public hash \PUH\ of the
contract from the code and the current time stamp and create a new
contractor at that address.  

We obtain the handle of the contract through a query, once the
contract is accepted on the blockchain in
\TirName{Block-Accept-Query}. The query's argument is the operation
hash, which is used to obtain the code and the time stamp of its
acceptance. From this information, we can re-calculate the public
hash. 

\subsection{Queries}
\label{sec:queries}

We conclude with two example transitions for a simple query in
Figure~\ref{fig:example-queries}. To obtain the balance of an implicit
account \PUK, we obtain the account info from the manager and extract
the balance (\TirName{Query-Balance-Implicit}). If the account is
unknown, then we raise an exception
(\TirName{Query-Balance-Fail}). Other queries are implemented analogously.

\section{Properties}
\label{sec:properties}
Having defined our execution model, we proceed to prove properties of the combined systems that ensure type-safe interaction between programs and the blockchain.
\subsection{Properties of blockchain state transitions}
One interesting property we wish to prove is that the execution of a program that starts with valid references to accounts, operations, and contracts is not corrupted by a transition.
\begin{proposition}
\label{propo:1}
  The following properties are preserved by a step on a well-formed
  configuration $ [\EXPRS, \ACCOUNTS] \| \BLOCKCHAIN$:
  \begin{itemize}
  \item for all $\OPH$ in $\EXPRS$, $\OPH \in \DOM (\BLOCKCHAIN.\PENDING)$,
  \item for all $\PUK$ in $\EXPRS$, $\PUK \in \DOM (\BLOCKCHAIN.\MANAGERS)$,
  \item for all $\PUH$ in $\EXPRS$, $\PUH \in \DOM (\BLOCKCHAIN.\CONTRACTORS)$.
  \end{itemize}
\end{proposition}

\begin{proposition}
\label{propo:2}
If $[\PENDING, \MANAGERS, \CONTRACTORS, \TIME] \BlockTrans{}
[\PENDING', \MANAGERS', \CONTRACTORS', \TIME']$, then
\begin{enumerate}
\item $\TIME \le \TIME'$
\item\label{item:1} $\DOM (\PENDING) \subseteq \DOM(\PENDING')$
\item invariant for the pool: if  $\PENDING (\OPH) = \Angle{\OP,
    \hat\TIME, \STATUS}$, then $\OPH = \GENERATEOPH (\OP,
  \hat\TIME)$. 
\item\label{item:7} for all $\OPH \in \DOM(P)$, if 
  $\PENDING (\OPH) = \Angle{\OP,  \hat\TIME, \STATUS}$, then either
  \begin{itemize}
  \item     $\PENDING'    (\OPH) = \PENDING (\OPH)$; or
  \item $\STATUS  = \STATUSPENDING$ and $\PENDING' (\OPH) =
    \Angle{\OP, \hat\TIME, \STATUSTIMEOUT}$; or
  \item  $\STATUS  = \STATUSPENDING$, $\TIME - \hat\TIME
    \le 60 $, 
    $\PENDING' (\OPH) =       \Angle{\OP, \hat\TIME, \STATUSINCLUDED\
      \TIME}$, and $\TIME'=\TIME+1$.
  \end{itemize}
 \item\label{item:2} for all $\OPH \in \DOM(P)$ and $\PENDING (\OPH) = \Angle{\OP,  \hat\TIME, \STATUS}$, 
    \begin{itemize}
    	\item  if $\STATUS=\STATUSPENDING$ and $\MANAGERS(\OP.\PUK) =  \Angle{\BAL, \COU}$ then  
  $\COU.b = \TRUE$ and $\BAL \ge \OP.\NTEZ + \OP.\MTEZ$;
\item if $\STATUS=\STATUSINCLUDED\ \hat\TIME$, then $\hat\TIME<\TIME'$.
    \end{itemize}
 \item\label{item:3} $ \DOM (\MANAGERS) \subseteq \DOM (\MANAGERS')$
 \item\label{item:4} for all $\PUK \in \DOM (\MANAGERS)$\\
   if $\MANAGERS (\PUK) =
   \Angle{\BAL, \COU}$, 
   then $\MANAGERS' (\PUK) =
   \Angle{\BAL', \COU'}$ and
   \begin{itemize}
   \item if $\COU.b=\TRUE$ and $\COU'.b=\FALSE$, then $\COU.n' \in \{
     \COU.n,  \COU.n+1\}$,
   \item otherwise $\COU.n = \COU'.n$
   \item If $\COU.n = \COU'.n$, then $\BAL = \BAL'$.
   \end{itemize}

 \item\label{item:5} $ \DOM (\CONTRACTORS) \subseteq \DOM (\CONTRACTORS')$
   \begin{itemize}
   \item for all $\PUH \in \DOM (\CONTRACTORS)$,
     $\CONTRACTORS (\PUH).\CODE = \CONTRACTORS' (\PUH).\CODE$
   \end{itemize}
 \item\label{item:6} invariant for contractors:
   for all $\PUH \in \DOM (\CONTRACTORS)$, \\
    $\CONTRACTORS (\PUH)$ = $\Angle{ \CODE, \tilde\TIME, \BAL, \STORAGE}$
   implies that $\PUH$ = $\GENERATEHASH (\CODE, \tilde\TIME)$.
\end{enumerate}
\end{proposition}
Establishing items~\ref{item:7} and~\ref{item:4} relies on the preimage resistance of the
various hash functions used to calculate operation hashes and public
hashes: we always feed a fresh timestamp into the hash
functions for operations and code. Items~\ref{item:1}--\ref{item:2} describe an invariant and
the lifecycle of 
operations. Items~\ref{item:3} and \ref{item:4} describe the lifecycle
of a transfer and items~\ref{item:5} and~\ref{item:6} describe
invariants for contractors. The invariants establish the
self-verifying property common of blockchain entities.

The proofs of these properties refer to all transitions with the detailed specifications of the related functions, such as chkCount and updSucc. Due to page limitations, not all transitions and their associated functions are presented in this paper, so the full proofs will be provided in an upcoming technical report. In this paper, we only provide the proofs for Proposition 2 at items 4 and 7.  

\begin{proof}[4]
 After feeding into a node, the status of the operation is \STATUSPENDING\ according to the transition \TirName{Node-Inject}. This operation could either be accepted by the blockchain on the condition that the elapsed time is less than 60 ($\TIME - \hat\TIME \le 60$), and then its status is \STATUSINCLUDED\ \TIME \ (the transition \TirName{Block-Accept}) or it is timed out with the  \STATUSTIMEOUT\ status (\TirName{Block-Timeout}). When an operation is accepted or timed out, its status is never changed. Therefore, if $\PENDING (\OPH)$ = $\Angle{\OP,  \hat\TIME, \STATUS}$, then there are  three cases:
 
 \begin{itemize}
 \item [(1)] if the operation's status remains the same as $\STATUS$ (still in \STATUSPENDING, \STATUSINCLUDED\ or \STATUSTIMEOUT), then we have $\PENDING' (\OPH)$ = $\Angle{\OP,  \hat\TIME, \STATUS}$. This means $\PENDING' (\OPH)$ = $\PENDING (\OPH)$; 
 \item [(2)] if the operation's status is \STATUSPENDING\ (\STATUS\ = \STATUSPENDING), and then the operation is timed out, then we have $\PENDING' (\OPH)$ = $\Angle{\OP, \hat\TIME, \STATUSTIMEOUT}$ according to the transition \TirName{Block-Timeout}; 
 \item [(3)] if the operation's status is \STATUSPENDING, the time condition is satisfied, and then the operation is accepted, then we have $\PENDING' (\OPH) = \Angle{\OP, \hat\TIME, \STATUSINCLUDED\ \TIME }$ and $\TIME'=\TIME +1$ because the timestamp is incremented by one according to the transition \TirName{Block-Accept}.
 \end{itemize}
From (1), (2) and (3), the item \ref{item:7} of Proposition 2 is proved.
\end{proof}

\begin{proof}[7]To prove this point, let us consider the two related functions. The function $\UPDATECOU(\MANAGERS, \PUK, \BOOLEAN)$ updates the flag of the counter of the account associated with the public key $\PUK$. Its specification is as follows:
 \begin{itemize}
   \item[] \UPDATECOU($\PUK \mapsto \langle \BAL, (\NAT, \hat\BOOLEAN) \rangle$, \BOOLEAN) = $\PUK \mapsto \langle \BAL, (\NAT, \BOOLEAN) \rangle $
 \end{itemize}
    
    The function $\UPDATESUCC(\MANAGERS, \PUK, \NTEZ, \MTEZ)$ updates the balance and the counter of the account associated with the public key $\PUK$. Its specification is as follows:
 \begin{itemize}
   \item[]  \UPDATESUCC($\PUK \mapsto \Angle{\BAL, (\NAT, \TRUE)}$, \NTEZ, \MTEZ) = $ \PUK \mapsto \Angle{\BAL - \NTEZ - \MTEZ, (\NAT + 1, \FALSE)}$    
 \end{itemize}
    
\noindent if $\MANAGERS (\PUK) =
   \Angle{\BAL, \COU}$, 
   then $\MANAGERS' (\PUK) =
   \Angle{\BAL', \COU'}$ and we have:
 \begin{itemize}
   \item[(1)] $\COU.b=\TRUE$ means that the operation is injected and its status is \STATUSPENDING\ at the time \TIME \ according to the transition \TirName{Node-Inject}.  After that, there are only two cases where the counter's flag is reset to False. If the operation is accepted, the counter's flag is reset ($\COU'.b=\FALSE$) according to the transition \TirName{Block-Accept} and the counter's value is incremented by 1 according to the specification of the function \UPDATESUCC\ ($\COU.n' =\COU.n+1$). In another case, if the operation is timed out, the counter's flag is also reset to False, but the value of the counter remains the same ($\COU.n' =\COU$) according to the transition \TirName{Block-timeout}. That is, if $\COU.b=\TRUE$ and $\COU'.b=\FALSE$, then $\COU.n' \in \{\COU.n, \COU.n+1\}$;
   \item[(2)] otherwise, if the operation is still \STATUSPENDING, the counter's value remains the same. This means $\COU.n = \COU'.n$; 
   \item[(3)] and then $\COU.n = \COU'.n$ means that the operation is either still \STATUSPENDING\ or  it has timed out. Therefore, the balance of the account remains the same because the balance is only changed when the operation is accepted. This means $\BAL = \BAL'$.
 \end{itemize}
 
 From (1), (2) and (3), the item \ref{item:4} of Proposition 2 is proved.
\end{proof}

\subsection{Typing related properties}
To describe the typing of contracts we maintain an environment
$\Delta ::= \EmptyEnv \mid \PUH: \TYPE, \Delta$ that associates a
public hash with a type.
We define typing for blockchains, local nodes, and configurations.
\begin{mathpar}
  \inferrule{
    \DOM (\Delta) = \DOM (\BLOCKCHAIN.\CONTRACTORS) \\
    (\forall\PUH\in\DOM (\Delta)) \\
    \Delta (\PUH) = \TPAIR\ \TYPE_p\ \TYPE_s \\
    \JTypeCode{\BLOCKCHAIN.\CONTRACTORS (\PUH).\CODE} {\TPAIR\ \TYPE_p\ \TYPE_s}
    \\
    \JTypeValue{\BLOCKCHAIN.\CONTRACTORS (\PUH).\STORAGE}{ \TYPE_s}
  }{\JTypeBlockchain\Delta\BLOCKCHAIN}
\end{mathpar}
The type for a hash is a pair type, which coincides with the type of
the code stored at that hash. The storage at that hash has the type
expected by the code.
\begin{mathpar}
  \inferrule{
    \JTypeExpr\cdot{ \EXPR_i}\TUNIT
  }{
    \JTypeNode{[\overline{\EXPR}, \ACCOUNTS]}
  }

  \inferrule{
    \JTypeBlockchain\Delta\BLOCKCHAIN \\
    \JTypeNode{\NODE_i}
  }{\JTypeConfig\Delta{\BLOCKCHAIN[\overline\NODE]}}
\end{mathpar}

\begin{lemma}[Preservation]
  If $\BLOCKCHAIN[\overline\NODE] \SystemTrans{} \BLOCKCHAIN'[\overline{\NODE'}]$ and
  $\JTypeConfig\Delta {\BLOCKCHAIN[\overline\NODE]}$, then there is some
  $\Delta' \supseteq \Delta$ such that
  $\JTypeConfig{\Delta'}{ \BLOCKCHAIN'[\overline{\NODE'}]}$.
\end{lemma}
This lemma includes the standard preservation for the lambda calculus part.
\begin{lemma}[Progress]
  If $\JTypeConfig\Delta{\BLOCKCHAIN[\overline\NODE]}$, then either
  all expressions in all nodes are unit values or there is a
  configuration $\BLOCKCHAIN'[\overline\NODE']$ such that
  $\BLOCKCHAIN[\overline\NODE] \SystemTrans \BLOCKCHAIN'[\overline\NODE']$.
\end{lemma}
The consistency lemma says that all committed transactions respect the
typing.
\begin{lemma}[Consistency]
Consider a blockchain state with $\JTypeBlockchain\Delta {[\PENDING, \MANAGERS, \CONTRACTORS,
\TIME]}$.

For all $\OPH\in\DOM (\PENDING)$,
if  $\PENDING (\OPH) =  \Angle{\OP, \hat\TIME, \STATUS}$
\begin{itemize}
\item if $\OP = \TRANSFER\NTEZ\PUK{\PUK'}\MTEZ$, then \\
  $\PUK, \PUK'$ $\in$ $\DOM (\MANAGERS)$;
\item if $\OP = \TRANSFER[\PARAMETER]\NTEZ\PUK\PUH\MTEZ$, then
  \begin{itemize}
  \item $\PUK \in \DOM (\MANAGERS)$ and $\PUH \in\DOM (\CONTRACTORS)$,
  \item $\JTypeValue\PARAMETER{ \TYPE_p}$ where $\Delta (\PUH)  =\TPAIR\ \TYPE_p\ \TYPE_s$;
  \end{itemize}
\item if $\OP = \ORIGINATE\NTEZ\PUK\CODE\STRING\MTEZ$ and $\STATUS = \STATUSINCLUDED\ \TIME'$, then
  \begin{itemize}
  \item $\PUK \in \DOM (\MANAGERS)$ and  $\PUH =\GENERATEHASH (\CODE, \TIME') \in \DOM (\CONTRACTORS)$,
  \item $\Delta (\PUH)  =\TPAIR\ \TYPE_p\ \TYPE_s$,  $\JTypeCode\CODE{
      \TPAIR\ \TYPE_p\ \TYPE_s}$ and $\JTypeValue\STRING{ \TYPE_s}$.
  \end{itemize}
\end{itemize}
\end{lemma}

\begin{proof} Consider the proof of the second item of Lemma 3, which specifies the property on type for a smart contract invocation. A smart contract call \OP\ has the form \TRANSFER [\PARAMETER]\NTEZ\PUK\PUH\MTEZ. If $\PENDING (\OPH) = \Angle{\OP, \hat\TIME, \STATUS }$, then the operation $\OP$ is injected into the node. According to the transition \TirName{Node-Inject} for a smart contract invocation, the public key is valid and the public hash must be associated with a smart contract on the blockchain. This means $\PUK \in \DOM (\MANAGERS)$ and $\PUH\in \DOM (\CONTRACTORS)$. Moreover, the chkArg function checks whether the type of the input parameter \PARAMETER\ matches the parameter type of the smart contract. If the casted type of the smart contract is $\Delta (\PUH) =\TPAIR\ \TYPE_p\ \TYPE_s$, then the type of the parameter  must be $\TYPE_p$. This means $\JTypeValue\PARAMETER {\TYPE_p}$. Therefore, this item is proved.
\end{proof}

\section{Related Work}
\label{sec:related work}
The inability to access external data sources limits the potential of
smart contracts. Oracles
\cite{oracle-patterns,call-action-oracle,oracles-study} can help
overcome this limitation by providing a bridge between the outside
sources and the blockchain network. A blockchain oracle is used to provide external data for smart contracts. When the external data is available, an oracle invokes a smart contract with that information. The invocation can conveniently be made through a programmatic interface. There has been extensive research on providing oracle solutions for blockchain.  Adler et al
\cite{blockchain-oracles} propose a framework to explain blockchain
oracles and various key aspects of oracles. This framework aims to
provide developers with a guide for incorporating oracles into
blockchain-based applications. The main problems with using a blockchain oracle are the untrusted
data provided maliciously or inaccurately \cite{trustworthy}. Ma at el
\cite{reliable-oracle} propose an oracle equipped with 
verification and disputation mechanisms. Similarly, Lo et al
\cite{reliablity-oracles} provide a framework for performing
reliability analysis of various blockchain oracle platforms.


Current blockchains such as Ethereum \cite{eth-whitepaper} and Tezos \cite{tezos-whitepaper} often offer RPC APIs and use loosely structured data, such as a JSON-based format that is difficult for a programmatic program to handle. As a result, there is increasing work to provide better programmatic interfaces to blockchains. Web3.js \cite{web3.js} provides an Ethereum JavaScript API and offers Java Script users a convenient interface to interact with the Etherum blockchain. Later, Web3.py \cite{web3.py}, derived from Web3.js, is developed to provide a Python library for interacting with Ethereum. Our typed API not only supports for programmatic programs, but also provides verifiable interaction with the Tezos smart contract platform.

\section{Conclusion}
\label{sec:conclusion}

We present a first step towards a typed API for smart contracts on the
Tezos blockchain. Our formalization enables us to establish basic
properties of the interaction between ordinary programs and smart
contracts. We see ample scope for future work to provide a
higher-level interface that exploits the similarities between
blockchain programming and concurrent programs. The next step will be to formalize the typing-related results. The formalization could connect with the Mi-Cho-Coq  formalization of Michelson contracts \cite{DBLP:conf/fm/BernardoCHPT19}. In the end, we would like to state and prove properties of a system that contains OCaml code (multi-threaded or distributed) connected to Michelson contracts on the Tezos blockchain via the typed API.

\bibliographystyle{splncs04}
\bibliography{bio}

\include{appendix}
\end{document}

%% file: macros.tex
\newcommand{\Angle}[1]{\langle#1\rangle}

\newcommand\SUNIT{\text{()}}
\newcommand{\TRUE}{\text{True}}
\newcommand{\FALSE}{\text{False}}
\newcommand{\NEG}{\neg}

\newcommand{\PAK}{\textbf{pak}}
\newcommand{\PUK}{\textbf{puk}}
\newcommand{\ADDR}{\textbf{addr}}

\newcommand{\PUH}{\textbf{puh}}
\newcommand{\CODE}{\textbf{code}}
\newcommand{\BAL}{\textbf{bal}}
\newcommand{\COU}{\textbf{cnt}}

\newcommand{\STORAGE}{\textbf{storage}}
\newcommand{\OP}{\textbf{op}}

\newcommand{\OPH}{\textbf{oph}}
\newcommand{\TIME}{\textbf{t}}
\newcommand{\CONTRACTORS}{\textbf{C}}
\newcommand{\PENDING}{\textbf{P}}
\newcommand{\MANAGERS}{\textbf{M}} 

\newcommand{\STATUSPENDING}{\text{pending}} 
\newcommand{\STATUSINCLUDING}{\text{included}}
\newcommand{\STATUSTIMEOUT}{\text{timeout}}
\newcommand{\STATUSINCLUDED}{\text{included}}

\newcommand{\TPAK}{\text{Pak}}
\newcommand{\TPUK}{\text{Puk}}
\newcommand{\TADDR}{\text{Addr}}
\newcommand{\TPUH}{\text{Puh}}
\newcommand{\TCODE}{\text{Code}}
\newcommand{\TBAL}{\text{Bal}}
\newcommand{\TCOU}{\text{Cnt}}
\newcommand{\TSTORAGE}{\text{Storage}}
\newcommand{\TOP}{\text{Op}}
\newcommand{\TOPH}{\text{Oph}}
\newcommand{\TTIME}{\text{Time}}
\newcommand{\TTEZ}{\text{Tz}}

\newcommand{\TRANSFER}[5][\SUNIT]{\text{transfer $#2$ from $#3$ to $#4$ arg $#1$ fee $#5$}}
\newcommand{\ORIGINATE}[5]{\text{originate contract transferring $#1$ from $#2$ running $#3$ init $#4$ fee $#5$}}
\newcommand{\NTEZ}{\textbf{nt}}
\newcommand{\MTEZ}{\textbf{fee}}

\newcommand\STRING{\textbf{s}}
\newcommand\BOOLEAN{\textbf{b}}
\newcommand\NAT{\textbf{n}}
\newcommand\INIT{\textbf{s}}
\newcommand\PARAMETER{\textbf{p}}

\newcommand{\GETBALANCE}[1]{\text{balance $#1$}}
\newcommand{\GETSTATUS}[1]{\text{status $#1$}}
\newcommand{\GETSTORAGE}[1]{\text{storage $#1$}} 

\newcommand{\GETCONTRACT}[1]{\text{contract $#1$}}

\newcommand{\ACCOUNTS}{\textbf{A}}

\newcommand\ECS{\textbf{EC}}
\newcommand\ECN[1][E]{\underline{\textbf{#1}}}
\newcommand\EC[2][E]{\underline{\textbf{#1}}[#2]}
\newcommand\ECHOLE{[~]}
\newcommand{\EXPRS}{\overline{\textbf{e}}}
\newcommand{\EXPR}{\textbf{e}}
\newcommand{\VARIABLE}{\textbf{x}}
\newcommand{\CONSTANT}{\textbf{c}}
\newcommand{\CONS}{\text{cons}}
\newcommand{\NIL}{\text{nil}}
\newcommand\LEFT{\text{left}}
\newcommand\RIGHT{\text{right}}
\newcommand{\AND}{\text{and}}
\newcommand{\OR}{\text{or}}
\newcommand{\NOT}{\text{not}}
\newcommand{\PLUS}{\textbf{+}}

\newcommand{\EQUAL}{\textbf{=}}

\newcommand\FIX{\text{fix}}

\newcommand\RAISE{\text{raise}}
\newcommand\TRY{\text{try}}
\newcommand\EXCEPT{\text{except}}
\newcommand\MATCH{\text{match}}
\newcommand\WITH{\text{with}}
\newcommand\PATTERN{\textbf{pat}}
\newcommand\SOME{\text{some}}
\newcommand\NONE{\text{none}}
\newcommand\VAL{\textbf{v}}
\newcommand\SC{\textbf{sc}}

\newcommand\CAST[3]{(#1 : #2 \Rightarrow #3)}

\newcommand\TYPE{T}
\newcommand\TYPEU{U}
\newcommand\TINT{\text{Int}}
\newcommand\TUNIT{\text{Unit}}
\newcommand\TBOOL{\text{Bool}}
\newcommand\TSTRING{\text{Str}}
\newcommand\TSTATUS{\text{Status}}
\newcommand\TPAIR{\text{Pair}}
\newcommand\TLIST{\text{List}}
\newcommand\TSUM{\text{Or}}
\newcommand\TOPTION{\text{Option}}
\newcommand\TEXCEPTION{\text{Exc}}
\newcommand\TCONTRACT{\text{Cont}}
\newcommand\TNO{\top}

\newcommand\QOP{\textbf{qop}}
\newcommand\INT{\textbf{i}}

\newcommand\TEnv{\Gamma}
\newcommand\JTypeCode[2]{\vdash_C#1 : #2}
\newcommand\JTypeValue[2]{\vdash_V#1:#2}
\newcommand\JTypeExpr[3]{#1 \vdash #2 : #3}
\newcommand\JTypeNode[1]{\vdash #1 \text{ ok}}
\newcommand\JTypeConfig[2]{#1 \vdash #2}
\newcommand\JTypeBlockchain[2]{#1 \vdash #2}
\newcommand\SubType{<:}

\newcommand{\NODE}{\textbf{N}}
\newcommand{\BLOCKCHAIN}{\textbf{B}}

\newcommand{\CHECKARG}{\textup{chkArg}}
\newcommand{\CHECKGAS}{\textup{chkFee}}

\newcommand{\CHECKBAL}{\textup{chkBal}}
\newcommand{\CHECKCOU}{\textup{chkCount}}

\newcommand{\CHECKPUH}{\textup{chkPuh}}
\newcommand{\CHECKPRG}{\textup{chkPrg}}

\newcommand{\CHECKINIT}{\textup{chkInit}}

\newcommand{\UPDATECOU}{\textup{updCount}}
\newcommand{\UPDATECONSTR}{\textup{updConstr}}
\newcommand{\UPDATESUCC}{\textup{updSucc}}

\newcommand{\GENERATEOPH}{\textup{genOpHash}}
\newcommand\GENERATEHASH{\textup{genHash}}

\newcommand{\ExprTrans}[1][{}]{\stackrel{#1}\longrightarrow_E}
\newcommand{\BlockTrans}[1][{}]{\stackrel{#1}\longrightarrow_B}
\newcommand{\NodeTrans}[1][{}]{\stackrel{#1}\longrightarrow_N}
\newcommand{\SystemTrans}{\longrightarrow}

\newcommand{\EmptyEnv}{\cdot}

\newcommand\STATUS{\textbf{st}}
\newcommand\ERROR{\textbf{err}}
\newcommand\ERRPRG{\text{xPrg}}
\newcommand\ERRBAL{\text{xBal}}
\newcommand\ERRCOUNT{\text{xCount}}
\newcommand\ERRFEE{\text{xFee}}
\newcommand\ERRPUK{\text{xPub}}
\newcommand\ERRPUH{\text{xPuh}}
\newcommand\ERRARG{\text{xArg}}
\newcommand\ERRINIT{\text{xInit}}

\newcommand\partialto\hookrightarrow
\newcommand\DOM{\textit{dom}}
\newcommand\Nat{\textbf{N}}


%% file: appendix.tex
\clearpage
\appendix

\section{Type Soundness}
\label{sec:type-soundness}

\begin{lemma}[Canonical forms]\label{lemma:canonical-forms}
  Given a set of local accounts $\ACCOUNTS$, a blockchain
  $\BLOCKCHAIN$, and a typed value $\JTypeExpr\cdot \VAL \TYPE$.
  \begin{itemize}
  \item If $\TYPE=\TPUH$, then $\VAL=\PUH$ and $\PUH \in \DOM
    (\BLOCKCHAIN.\CONTRACTORS)$.
  \item If $\TYPE = \TPUK$, then $\VAL=\PUK$ and $\exists\PAK$ such
    that $(\PAK, \PUK) \in \ACCOUNTS$.
  \item If $\TYPE = \TADDR$, then $\VAL$ is $\PUH$ or $\PUK$.
  \item If $\TYPE = \TCONTRACT\ \TYPE_p\ \TYPE_s$, then $\VAL=\PUH$
    and $\PUH \in \DOM (\BLOCKCHAIN.\CONTRACTORS)$ and
    $\BLOCKCHAIN.\CONTRACTORS (\PUH) = (\CODE, \TIME, \BAL, \STORAGE)$
    such that $      \JTypeCode \CODE{ \TPAIR\ \TYPE_p\ \TYPE_s}$ and
    $\JTypeValue \STORAGE { \TYPE_s}$.
  \item If $\TYPE = \TCODE\ \TYPE\ \TYPEU$, then $\VAL = \CODE$ and
    $\JTypeCode \CODE{\TPAIR\ \TYPE\ \TYPEU}$.
  \item If $\TYPE = \TOPH\ \TYPE\ \TYPEU$, then $\VAL = \OPH$ and
    $\OPH\in \DOM (\BLOCKCHAIN.\PENDING)$ and $\BLOCKCHAIN.\PENDING
    (\OPH) = \Angle{\OP, \TIME, \STATUS}$ where $\TYPE=\TYPEU=\TNO$ if
    $\OP$ is a transfer and $\TYPE=\TYPE_p\ne\TNO$, $\TYPEU=\TYPE_s\ne\TNO$ if $\OP
    = \ORIGINATE\NTEZ\PUK\CODE\STRING\MTEZ$ and $\JTypeCode\CODE{
      \TPAIR\ \TYPE_p\ \TYPE_s}$.
  \item If $\TYPE = \TSTATUS$, then $\VAL \in \{\STATUSPENDING , \STATUSINCLUDING (\INT),
    \STATUSTIMEOUT\}$.
  \item If $\TYPE = \TEXCEPTION$, then $\VAL \in \{\ERRPRG, \ERRBAL, \ERRCOUNT, \ERRFEE,
    \ERRPUK, \ERRPUH, \ERRARG, \ERRINIT \}$.
  \item If $\TYPE = \TTEZ$, then $\VAL = \NTEZ$, a token amount.
  \item If $\TYPE = \TNO$, then $\VAL$ can be any syntactic value.
  \item If $\TYPE = \TINT$, then $\VAL = \INT$.
  \item If $\TYPE = \TUNIT$, then $\VAL = \SUNIT$.
  \item If $\TYPE = \TBOOL$, then $\VAL \in \{ \TRUE, \FALSE\}$.
  \item If $\TYPE = \TSTRING$, then $\VAL = \STRING$, a string.
  \item If $\TYPE = \TYPE\to\TYPEU$, then $\VAL = \lambda
    \VARIABLE. \EXPR$. 
  \item If $\TYPE = \TPAIR\ \TYPE\ \TYPEU$, then $\VAL = (\VAL',
    \VAL'')$ where $\JTypeExpr\cdot {\VAL'}\TYPE$ and $\JTypeExpr\cdot
    {\VAL''}\TYPEU$ in context $\ACCOUNTS$ and $\BLOCKCHAIN$.
  \item If $\TYPE = \TLIST\ \TYPE$, then either $\VAL = \NIL$ or $\VAL
    = \CONS\ {\VAL'}\ {\VAL''}$ where $\JTypeExpr\cdot{\VAL'}\TYPE$
    and $\JTypeExpr\cdot{\VAL''}{\TLIST\ \TYPE}$ in context
    $\ACCOUNTS$ and $\BLOCKCHAIN$.
  \item If $\TYPE = \TSUM\ \TYPE\ \TYPEU$, then either $\VAL = \LEFT\
    \VAL'$ where $\JTypeExpr\cdot {\VAL'}\TYPE$ or $\VAL= \RIGHT\
    \VAL''$ where $\JTypeExpr\cdot {\VAL''}\TYPEU$ in context
    $\ACCOUNTS$ and $\BLOCKCHAIN$.
  \item If $\TYPE = \TOPTION\ \TYPE$, then $\VAL = \NONE$ or $\VAL =
    \SOME\ \VAL'$ where $\JTypeExpr\cdot {\VAL'}\TYPE$ in context
    $\ACCOUNTS$ and $\BLOCKCHAIN$.
  \end{itemize}
\end{lemma}

\begin{lemma}[Subterm replacement]\label{lemma:subterm-replacement}
  If $\JTypeExpr\cdot{ \EC\EXPR}\TYPE$,
  $\JTypeExpr\cdot\EXPR \TYPE'$, and
  $\JTypeExpr\cdot{\EXPR'} \TYPE'$, then
  $\JTypeExpr\cdot{ \EC{\EXPR'}}\TYPE$.
\end{lemma}
\begin{proof}
  Induction on evaluation context $\ECN$ making use of the fact that
  an evaluation context does not bind variables.
\end{proof}

\begin{lemma}[Preservation for expressions]\label{lemma:type-preservation-expressions}
  If $\JTypeExpr\cdot{ \EXPR}\TYPE$ and $\EXPR \ExprTrans \EXPR'$,
  then  $\JTypeExpr\cdot{ \EXPR'}\TYPE$.
\end{lemma}
\begin{proof}
  Standard result: type preservation for simply typed lambda calculus
  with pairs, sums, and exceptions. Uses
  Lemma~\ref{lemma:subterm-replacement} for reductions in evaluation
  context. See, for instance, Types in
  Programming Languages by Benjamin Pierce.
\end{proof}

\begin{lemma}[Preservation]
  If $\BLOCKCHAIN[\overline\NODE] \SystemTrans{} \BLOCKCHAIN'[\overline{\NODE'}]$ and
  $\JTypeConfig\Delta {\BLOCKCHAIN[\overline\NODE]}$, then there is some
  $\Delta' \supseteq \Delta$ such that
  $\JTypeConfig{\Delta'}{ \BLOCKCHAIN'[\overline{\NODE'}]}$.
\end{lemma}
\begin{proof}
  The proof is by induction on the reduction relation
  $\BLOCKCHAIN[\overline\NODE] \SystemTrans{}
  \BLOCKCHAIN'[\overline\NODE']$ and inversion of the typing
  judgments. We only consider the exemplary reductions shown in the
  paper. We mark all components that belong to the reductum with 
  $'$ as in $\NODE'$.

  From $\JTypeConfig\Delta {\BLOCKCHAIN[\overline\NODE]}$ we obtain
  \begin{gather}
    \label{eq:1}
    \JTypeBlockchain\Delta\BLOCKCHAIN
    \\\label{eq:2}
    \JTypeNode{\NODE_i}
  \end{gather}
  From~\eqref{eq:1} we obtain
  \begin{gather}
    \label{eq:3}
    \DOM (\Delta) = \DOM (\BLOCKCHAIN.\CONTRACTORS)
    \intertext{and $\forall\PUH\in\DOM (\Delta)$}\label{eq:4}
    \Delta (\PUH) = \TPAIR\ \TYPE_p\ \TYPE_s
    \\\label{eq:5}
    \JTypeCode{\BLOCKCHAIN.\CONTRACTORS (\PUH).\CODE} {\TPAIR\ \TYPE_p\ \TYPE_s}
    \\
    \JTypeValue{\BLOCKCHAIN.\CONTRACTORS (\PUH).\STORAGE}{ \TYPE_s}
  \end{gather}
  From~\eqref{eq:2} we obtain, if $\NODE_i = [\overline{\EXPR_i},
  \ACCOUNTS_i]$, 
  \begin{gather}
    \label{eq:6}
    \JTypeExpr\cdot{ \EXPR_{ij}}\TUNIT
  \end{gather}

  \textbf{Reduction }$  \inferrule[Config-Node]
  {\NODE_0 \NodeTrans \NODE_0'}
  { {\BLOCKCHAIN[\NODE_0 :: \overline\NODE]}
    \SystemTrans
    {\BLOCKCHAIN[\NODE_0' :: \overline\NODE]}}$.

  The only possible reduction here is $  \inferrule[Node-Eval]
  {
    \EXPR \ExprTrans \EXPR'
  }{
    [\EC\EXPR :: \EXPRS, \ACCOUNTS] \NodeTrans{}
    [\EC{\EXPR'} :: \EXPRS, \ACCOUNTS]
  }$.

  From~\eqref{eq:6}, we know that $\JTypeExpr\cdot{ \EXPR}\TUNIT$. By
  Lemma~\ref{lemma:type-preservation-expressions}, $\JTypeExpr\cdot{
    \EXPR'}\TUNIT$, the types of the $\EXPRS$ are not affected, hence
  $\JTypeNode{\NODE_0' = [\EC{\EXPR'} :: \EXPRS, \ACCOUNTS]}$. None of the other
  nodes changed, neither did $\BLOCKCHAIN$, so that
  $\JTypeConfig\Delta{{\BLOCKCHAIN[\NODE_0' :: \overline\NODE]}}$. 

  \clearpage
  \textbf{Reduction }$\inferrule[Config-System]{
    \NODE\|\BLOCKCHAIN \SystemTrans \NODE'\|\BLOCKCHAIN'
}{
    \BLOCKCHAIN[\NODE :: \overline{\NODE}] \SystemTrans
    \BLOCKCHAIN'[\NODE' :: \overline{\NODE}]
  }$.

  We need to consider the cases for $\SystemTrans$.

  \bigskip\textbf{Subcase }$  \inferrule[Node-Inject]{
    \Angle{\PAK,\PUK} \in \ACCOUNTS \\
    \CHECKBAL (\MANAGERS, \PUK, \NTEZ, \MTEZ) \\
    \CHECKARG (\CONTRACTORS, \PUH, \PARAMETER) \\
    \CHECKCOU (\MANAGERS, \PUK) \\
    \CHECKPUH (\CONTRACTORS, \PUH) \\
    \CHECKGAS (\CONTRACTORS, \PUH, \PARAMETER, \MTEZ) \\
    \OPH = \GENERATEOPH (\OP, \TIME) \\
    \OP = \TRANSFER[\PARAMETER]\NTEZ\PUK{\PUH}\MTEZ    
  }{
    { [\EC\OP :: \EXPRS, \ACCOUNTS] \|
      [\PENDING, \MANAGERS, \CONTRACTORS, \TIME] } \SystemTrans \\
    { [\EC{\OPH}  :: \EXPRS, \ACCOUNTS] \|
      [ \OPH \mapsto \Angle{\OP, \TIME, \STATUSPENDING}
      ::\PENDING,
      \UPDATECOU(\MANAGERS, \PUK, \TRUE),
      \CONTRACTORS,
      \TIME]
    }
  }$.
  Here $\NODE = [\EC\OP :: \EXPRS, \ACCOUNTS]$. We first check type
  preservation for the expression part. There are two typing rules for
  the transfer
  $\OP$, but only the one for contract invocation applies as the other
  one requires $\JTypeExpr\cdot{\PUH}\TPUK$, which does not hold.


  For a contract invocation (specialized to empty environment)
  \begin{gather}
    \label{eq:8}
        \inferrule{
      \JTypeExpr\cdot{\NTEZ}\TTEZ \\
      \JTypeExpr\cdot{\PUK}\TPUK \\
      \JTypeExpr\cdot{\PUH}\TCONTRACT\ \TYPE_p\ \TYPE_s \\
      \JTypeExpr\cdot{\PARAMETER}{\TYPE_p} \\
      \JTypeExpr\cdot{\MTEZ}\TTEZ }{
      \JTypeExpr\cdot{\TRANSFER[\PARAMETER]{\NTEZ}{\PUK}{\PUH}{\MTEZ}}{\TOPH\
      \TNO\ \TNO}
    }
  \end{gather}
  The canoncical forms lemma~\ref{lemma:canonical-forms} is parameterized over the accounts
  $\ACCOUNTS$ of the local node and the current contractors
  $\BLOCKCHAIN.\CONTRACTORS$. Hence, we know that the arguments are
  legal, which is also checked by the rule.

  The reduct returns an operation hash $\OPH$ at type $\TOPH\ \TNO\
  \TNO$, which places no restrictions on the context of $\OPH$.

  Moreover, $\Delta' = \Delta$ and $\CONTRACTORS' = \CONTRACTORS$ as
  no new contract is originated.

  We conclude with Lemma~\ref{lemma:subterm-replacement} and
  reapplying \TirName{Config-System}.

  \bigskip\textbf{Subcase }$\inferrule[Contract-Yes]{
    \JTypeCode\CODE{\TPAIR\ \TYPE\ \TYPEU} \\
    \BLOCKCHAIN.\CONTRACTORS (\PUH) =  \Angle{\CODE, \tilde\TIME, \NTEZ', \STRING'}
  }{
    { [\EC{\CAST\PUH\TPUH{\TCONTRACT\ \TYPE}}  :: \EXPRS, \ACCOUNTS] \|
      \BLOCKCHAIN
    }
    \SystemTrans
    { [\EC{\PUH} :: \EXPRS, \ACCOUNTS] \|
      \BLOCKCHAIN
    }
  }$.

  Immediate using Lemma~\ref{lemma:canonical-forms} and
  Lemma~\ref{lemma:subterm-replacement}.

  \bigskip\textbf{Subcase }$\inferrule[Contract-No]{
    \BLOCKCHAIN.\CONTRACTORS (\PUH) = \Angle{\CODE, \tilde\TIME,
      \NTEZ', \STRING'} \Rightarrow {}
    \JTypeCode\CODE{\TPAIR\ \TYPE'\ \TYPEU} \wedge \TYPE \ne \TYPE' \\
  }{
    { [\EC{\CAST\PUH\TPUH{\TCONTRACT\ \TYPE}}  :: \EXPRS, \ACCOUNTS] \|
      \BLOCKCHAIN
    }
    \SystemTrans
    { [\EC{\RAISE\ \ERRPRG} :: \EXPRS, \ACCOUNTS] \|
      \BLOCKCHAIN
    }
  }$.

  The typing rule for $\RAISE$ can return any type. Hence, this is
  immediate by Lemma~\ref{lemma:subterm-replacement}.

  \bigskip\textbf{Subcase }$\inferrule[Block-Originate]{
    \Angle{\PAK,\PUK} \in \ACCOUNTS \\ \CHECKBAL (\MANAGERS, \PUK, \NTEZ, \MTEZ) \\
    \CHECKCOU (\MANAGERS, \PUK) \\
    \CHECKPRG (\CODE) \\
    \CHECKGAS (\CODE, \INIT, \NTEZ, \MTEZ)  \\
    \CHECKINIT (\CODE, \STRING) \\
    \OPH = \GENERATEOPH(\OP, \TIME) \\
    \OP = \ORIGINATE\NTEZ\PUK\CODE\STRING\MTEZ }{ [\EC\OP :: \EXPRS,
    \ACCOUNTS
    ] \| [\PENDING, \MANAGERS, \CONTRACTORS, \TIME] \SystemTrans \\
    [\EC{\OPH} :: \EXPRS, \ACCOUNTS] \| [\OPH \mapsto \Angle{\OP,
      \TIME, \STATUSPENDING} ::\PENDING,
    \UPDATECOU(\MANAGERS,\PUK,\TRUE), \CONTRACTORS, \TIME] }$. 
  
  Suppose that $\JTypeCode\CODE {\TPAIR\ \TYPE_p\ \TYPE_s}$. Then
  $\JTypeExpr\cdot\OP{\TOPH\ \TYPE_p\ \TYPE_s}$. But this is the type
  of the $\OPH$ in the reductum as it points to $\OP$ in
  $\PENDING$. Hence, the result is immediate by
  Lemma~\ref{lemma:subterm-replacement}.


  \bigskip\textbf{Subcase }$\inferrule[Query-Balance-Implicit]{
    \BLOCKCHAIN.\MANAGERS (\PUK) = \Angle{\BAL,\COU}
  }{[\EC{\GETBALANCE\PUK} :: \EXPRS, \ACCOUNTS] \| \BLOCKCHAIN
    \SystemTrans\ [\EC{\BAL} ::\EXPRS, \ACCOUNTS] \| \BLOCKCHAIN} $.
  
  The reduction replaces $\GETBALANCE\PUK$ of type $\TTEZ$ by $\BAL$
  of the same type. Hence, the result is immediate by
  Lemma~\ref{lemma:subterm-replacement}.

  \bigskip\textbf{Subcase }$\inferrule[Query-Balance-Fail]{ \PUK \notin \DOM
    (\BLOCKCHAIN.\MANAGERS) \ }{[\EC{\GETBALANCE\PUK} :: \EXPRS,
    \ACCOUNTS] \| \BLOCKCHAIN \SystemTrans {[\EC{\RAISE\ \ERRPUK}
      ::\EXPRS, \ACCOUNTS] \| \BLOCKCHAIN}}$.

  Immediate by Lemma~\ref{lemma:subterm-replacement} because $\RAISE$
  can have any type.
    
  \clearpage
  \textbf{Reduction }$\inferrule[Config-Block]
  {\BLOCKCHAIN \BlockTrans \BLOCKCHAIN'}
  { \BLOCKCHAIN[{\overline\NODE}]
    \SystemTrans
    \BLOCKCHAIN'[{\overline\NODE}]}$.

  We need to considere cases for $\BlockTrans$.

  \bigskip\textbf{Subcase }$\inferrule[Block-Accept]{
    \OP = \TRANSFER[\PARAMETER]\NTEZ\PUK{\PUH}\MTEZ \\
    \TIME - \hat\TIME \le 60
  }{
    { 
      [\OPH \mapsto \Angle{\OP, \hat \TIME, \STATUSPENDING}
      ::\PENDING, \MANAGERS,
      \CONTRACTORS, \TIME]}
    \BlockTrans 
    {
      [\OPH \mapsto \Angle{\OP, \hat\TIME, \STATUSINCLUDING\ \TIME} :: \PENDING}, \\
    { \UPDATESUCC (\MANAGERS, \PUK, \NTEZ, \MTEZ), 
      \UPDATECONSTR (\CONTRACTORS, \PUH, \NTEZ, \PARAMETER), \TIME +1]
    }
  }$.

  No typing-related properties are affected.

  \bigskip\textbf{Subcase }$\inferrule[Block-Originate-Accept]{
    \OP = \ORIGINATE\NTEZ\PUK\CODE\STRING\MTEZ \\
    \PUH = \GENERATEHASH(\CODE, \TIME) \\
    \TIME-\hat\TIME  \le 60
  }{
    [\OPH \mapsto \Angle{\OP, \hat\TIME, \STATUSPENDING} :: \PENDING, \MANAGERS, \CONTRACTORS, \TIME]
    \BlockTrans \\
    [\OPH \mapsto \Angle{\OP, \hat\TIME, \STATUSINCLUDING\ \TIME} :: \PENDING, \UPDATESUCC
    (\MANAGERS, \PUK, \NTEZ, \MTEZ),\\ \PUH \mapsto  \Angle{\CODE, \TIME, \NTEZ, \STRING} :: \CONTRACTORS, \TIME+1]
  }$.
  
  This reduction extends $\CONTRACTORS$ with a new entry for
  $\PUH$. To preserve typing, we need to extend $\Delta$ with the
  binding $\PUH : \TPAIR\ \TYPE_p\ \TYPE_s$ where $\JTypeCode\CODE
  {\TPAIR\ \TYPE_p\ \TYPE_s}$. The generated code pointer is obtained
  with a query operation via the operation hash $\OPH$, which is also
  connected to the parameter and storage types. 

  \bigskip\textbf{Subcase }$  \inferrule[Block-Timeout]{
    \TIME-\hat\TIME > 60
  }{ 
    {[\OPH \mapsto \Angle{\OP, \hat \TIME, \STATUSPENDING}
     ::\PENDING, \MANAGERS,
      \CONTRACTORS, \TIME]}
    \BlockTrans \\
    { 
      [\OPH \mapsto \Angle{\OP, \hat \TIME, \STATUSTIMEOUT}
     :: \PENDING,  \UPDATECOU(\MANAGERS, \OP.\PUK, \FALSE),
      \CONTRACTORS, \TIME]}
  }$.

  No typing-related properties are affected.
\end{proof}

\clearpage
\begin{lemma}[Progress for expressions]\label{lemma:progress-expressions}
  If $\JTypeExpr\cdot{ \EXPR}\TYPE$, then either
  \begin{itemize}
  \item $\EXPR$ is a value,
  \item $\EXPR \ExprTrans \EXPR'$, or
  \item $\EXPR = \EC{\EXPR'}$ is a blockchain operation in an evaluation context:
    \begin{itemize}
    \item $\EXPR' = \TRANSFER[\PARAMETER]\NTEZ\PUK{\PUH}\MTEZ$;
    \item $\EXPR' = \ORIGINATE\NTEZ\PUK\CODE\INIT\MTEZ$;
    \item $\EXPR' = \QOP\ \VAL$;
    \item $\EXPR' = \CAST\VAL\TYPE\TYPEU$ where $\TYPEU \SubType \TYPE$.
    \end{itemize}
  \end{itemize}
\end{lemma}
\begin{proof}
  Standard result: progress for simply type lambda calculus with
  pairs, sums, and exceptions. Upcasts are resolved by identity
  reductions. The blockchain operations including downcasts are not
  handled by the $\ExprTrans$ relation. 
\end{proof}
\begin{lemma}[Progress]
  If $\JTypeConfig\Delta{\BLOCKCHAIN[\overline\NODE]}$, then either
  all expressions in all nodes are unit values or there is a
  configuration $\BLOCKCHAIN'[\overline\NODE']$ such that
  $\BLOCKCHAIN[\overline\NODE] \SystemTrans \BLOCKCHAIN'[\overline\NODE']$.
\end{lemma}
\begin{proof}
  From $\JTypeConfig\Delta {\BLOCKCHAIN[\overline\NODE]}$ we obtain
  \begin{gather}
    \label{eq:101}
    \JTypeBlockchain\Delta\BLOCKCHAIN
    \\\label{eq:102}
    \JTypeNode{\NODE_i}
  \end{gather}
  From~\eqref{eq:101} we obtain
  \begin{gather}
    \label{eq:103}
    \DOM (\Delta) = \DOM (\BLOCKCHAIN.\CONTRACTORS)
    \intertext{and $\forall\PUH\in\DOM (\Delta)$}\label{eq:104}
    \Delta (\PUH) = \TPAIR\ \TYPE_p\ \TYPE_s
    \\\label{eq:105}
    \JTypeCode{\BLOCKCHAIN.\CONTRACTORS (\PUH).\CODE} {\TPAIR\ \TYPE_p\ \TYPE_s}
    \\
    \JTypeValue{\BLOCKCHAIN.\CONTRACTORS (\PUH).\STORAGE}{ \TYPE_s}
  \end{gather}
  From~\eqref{eq:102} we obtain, if $\NODE_i = [\overline{\EXPR_i},
  \ACCOUNTS_i]$, 
  \begin{gather}
    \label{eq:106}
    \JTypeExpr\cdot{ \EXPR_{ij}}\TUNIT
  \end{gather}

  For each such $\EXPR_{ij}$, Lemma~\ref{lemma:progress-expressions}
  yields that either
  \begin{itemize}
  \item $\EXPR_{ij}$ is a value; as it has type $\TUNIT$, we obtain
    $\EXPR_{ij} = \SUNIT$ by Lemma~\ref{lemma:canonical-forms};
  \item $\EXPR_{ij} \ExprTrans \EXPR_{ij}'$, in which case the whole
    system makes a step; or
  \item $\EXPR_{ij} = \EC{\EXPR}$ where $\EXPR$ is a blockchain operation.
  \end{itemize}

  \bigskip\textbf{Subcase } $\EXPR =
  \TRANSFER[\PARAMETER]\NTEZ\PUK{\PUH}\MTEZ$. In this case, the
  \TirName{Node-Inject} reduction is in principle enabled:
  \begin{mathpar}
    \inferrule[Node-Inject]{
      \Angle{\PAK,\PUK} \in \ACCOUNTS \\
      \CHECKBAL (\MANAGERS, \PUK, \NTEZ, \MTEZ) \\
      \CHECKARG (\CONTRACTORS, \PUH, \PARAMETER) \\
      \CHECKCOU (\MANAGERS, \PUK) \\
      \CHECKPUH (\CONTRACTORS, \PUH) \\
      \CHECKGAS (\CONTRACTORS, \PUH, \PARAMETER, \MTEZ) \\
      \OPH = \GENERATEOPH (\OP, \TIME) \\
      \OP = \TRANSFER[\PARAMETER]\NTEZ\PUK{\PUH}\MTEZ    
    }{
      { [\EC\OP :: \EXPRS, \ACCOUNTS] \|
        [\PENDING, \MANAGERS, \CONTRACTORS, \TIME] } \SystemTrans \\
      { [\EC{\OPH}  :: \EXPRS, \ACCOUNTS] \|
        [ \OPH \mapsto \Angle{\OP, \TIME, \STATUSPENDING}
        ::\PENDING,
        \UPDATECOU(\MANAGERS, \PUK, \TRUE),
        \CONTRACTORS,
        \TIME]
      }
    }
  \end{mathpar}
  Thanks to the canonical forms Lemma~\ref{lemma:canonical-forms}, we
  know that $\Angle{\PAK,\PUK} \in \ACCOUNTS$, $\CHECKARG
  (\CONTRACTORS, \PUH, \PARAMETER)$ holds, and $\CHECKPUH
  (\CONTRACTORS, \PUH)$ holds. If one of the remaining checks fails,
  then one of the \TirName{Node-Reject} transitions throws an
  exception, so the configuration steps in every case. 

  \bigskip\textbf{Subcase } $\EXPR = \ORIGINATE\NTEZ\PUK\CODE\INIT\MTEZ$. In
  this case, the \TirName{Block-Originate} reduction is in principle
  enabled:
  \begin{mathpar}
    \inferrule[Block-Originate]{
      \Angle{\PAK,\PUK} \in \ACCOUNTS \\ \CHECKBAL (\MANAGERS, \PUK, \NTEZ, \MTEZ) \\
      \CHECKCOU (\MANAGERS, \PUK) \\
      \CHECKPRG (\CODE) \\
      \CHECKGAS (\CODE, \INIT, \NTEZ, \MTEZ)  \\
      \CHECKINIT (\CODE, \STRING) \\
      \OPH = \GENERATEOPH(\OP, \TIME) \\
      \OP = \ORIGINATE\NTEZ\PUK\CODE\STRING\MTEZ }{ [\EC\OP :: \EXPRS,
      \ACCOUNTS
      ] \| [\PENDING, \MANAGERS, \CONTRACTORS, \TIME] \SystemTrans \\
      [\EC{\OPH} :: \EXPRS, \ACCOUNTS] \| [\OPH \mapsto \Angle{\OP,
        \TIME, \STATUSPENDING} ::\PENDING,
      \UPDATECOU(\MANAGERS,\PUK,\TRUE), \CONTRACTORS, \TIME] }
  \end{mathpar}
  Thanks to the canonical forms Lemma~\ref{lemma:canonical-forms}, we
  know that $\Angle{\PAK,\PUK} \in \ACCOUNTS$, $\CHECKPRG (\CODE)$ holds, and $\CHECKINIT
  (\CODE, \STRING)$ holds. If one of the remaining checks fails,
  then one of the \TirName{Node-Reject} transitions throws an
  exception, so the configuration steps in every case. 
  
  \bigskip\textbf{Subcase } $\EXPR = \QOP\ \VAL$. If $\EXPR =
  \GETBALANCE\VAL$, then inversion tells us that $\JTypeExpr\cdot{
    \VAL}\TADDR$ and by canonical forms
  (Lemma~\ref{lemma:canonical-forms}), it must be that $\VAL$ has the
  form $\PUK$ or $\PUH$. In any case, the value is a meaningful
  address for the manager $\MANAGERS$. Depending on whether the
  address is in use, one of the reductions
  \TirName{Query-Balance-Implicit} or \TirName{Query-Balance-Fail} can
  execute. There are further analogous reductions handling the case where
  $\VAL=\PUH$ and we ask for the balance of a smart contract.

  Most queries behave like $\GETBALANCE\cdot$, except getting a
  contract handle from an operation hash:

  \bigskip\textbf{Subcase }$\EXPR = \GETCONTRACT\VAL$. This query is somewhat
  special as it is handled with reduction
  \TirName{Block-Accept-Query}. By inversion and canonical forms
  (Lemma~\ref{lemma:canonical-forms}) we know that $\VAL = \OPH$ is a
  valid operation hash of type $\TCONTRACT\ \TYPE\ \TYPEU$ where
  $\TYPE\ne\TNO$ and $\TYPEU \ne\TNO$.

  However, this reduction is conditional on the state of the
  transaction; it requires the new contract to have status
  $\STATUSINCLUDED$. If the contract has status $\STATUSTIMEOUT$, then
  the query raises and exception, analogous to the
  \TirName{Query-Balance-Fail} reduction. If the contract has status
  $\STATUSPENDING$, then the expression is blocked, but the system can
  make a step using \TirName{Block-Originate-Accept} that changes the
  status from $\STATUSPENDING$ to $\STATUSINCLUDED$. Alternatively,
  \TirName{Block-Timeout} can make a step to change the status to
  $\STATUSTIMEOUT$. In any case, the system as a whole can make a reduction.

  \bigskip\textbf{Subcase } $\EXPR = \CAST\VAL\TYPE\TYPEU$ where $\TYPEU \SubType \TYPE$.
  As an example, we consider the reductions \TirName{Contract-Yes} and
  \TirName{Contract-No}, where a cast is applied to a value of type
  $\TPUH$. By canonical forms, we know that the value has the
  form $\PUH \in \DOM (\BLOCKCHAIN.\CONTRACTORS)$. The code pointed to
  by this hash is checked at run time and results either in a $\PUH$
  at suitable contract type (-Yes reduction) or in raising an
  exception (-No reduction). 

\end{proof}


%% file: paper.bbl
\begin{thebibliography}{10}
\providecommand{\url}[1]{\texttt{#1}}
\providecommand{\urlprefix}{URL }
\providecommand{\doi}[1]{https://doi.org/#1}

\bibitem{trustworthy}
Al-Breiki, H., Rehman, M.H.U., Salah, K., Svetinovic, D.: Trustworthy
  blockchain oracles: Review, comparison, and open research challenges. IEEE
  Access  \textbf{8},  85675--85685 (2020)

\bibitem{tezos-intropaper}
{Allombert}, V., {Bourgoin}, M., {Tesson}, J.: Introduction to the {Tezos}
  blockchain. In: 2019 International Conference on High Performance Computing
  Simulation (HPCS). pp. 1--10 (2019). \doi{10.1109/HPCS48598.2019.9188227}

\bibitem{oracles-study}
Beniiche, A.: A study of blockchain oracles (2020)

\bibitem{DBLP:conf/fm/BernardoCHPT19}
Bernardo, B., Cauderlier, R., Hu, Z., Pesin, B., Tesson, J.: {Mi-Cho-Coq}, a
  framework for certifying {Tezos} smart contracts. In: Formal Methods. {FM}
  2019 International Workshops, Revised Selected Papers, Part {I}. Lecture
  Notes in Computer Science, vol. 12232, pp. 368--379. Springer (2019).
  \doi{10.1007/978-3-030-54994-7\_28},
  \url{https://doi.org/10.1007/978-3-030-54994-7\_28}

\bibitem{eth-whitepaper}
Buterin, V.: A next-generation smart contract and decentralized application
  platform (2013), \url{https://ethereum.org/en/whitepaper/}

\bibitem{call-action-oracle}
Caldarelli, G.: Understanding the blockchain oracle problem: A call for action.
  Information  \textbf{11}(11) (2020)

\bibitem{ethereum-rpc}
{Ethereum JSON-RPC API} (2021),
  \url{https://ethereum.org/en/developers/docs/apis/json-rpc/}

\bibitem{tezos-whitepaper}
Goodman, L.: {Tezos}-a self-amending crypto-ledger (2014),
  \url{https://www.tezos.com/static/papers/white-paper.pdf}

\bibitem{reliablity-oracles}
Lo, S.K., Xu, X., Staples, M., Yao, L.: Reliability analysis for blockchain
  oracles. Computers {\&} Electrical Engineering  \textbf{83},  106582 (2020)

\bibitem{reliable-oracle}
{Ma}, L., {Kaneko}, K., {Sharma}, S., {Sakurai}, K.: Reliable decentralized
  oracle with mechanisms for verification and disputation. In: 2019 Seventh
  International Symposium on Computing and Networking Workshops (CANDARW). pp.
  346--352 (2019)

\bibitem{blockchain-oracles}
Mammadzada, K., Iqbal, M., Milani, F., García-Bañuelos, L., Matulevičius,
  R.: Blockchain Oracles: A Framework for Blockchain-Based Applications, pp.
  19--34. Springer Verlag (09 2020)

\bibitem{michelson}
{Michelson}: The language of smart contracts in {Tezos},
  \url{https://tezos.gitlab.io/alpha/michelson.html}

\bibitem{oracle-patterns}
Mühlberger, R., Bachhofner, S., Castelló~Ferrer, E., Di~Ciccio, C., Weber,
  I., Wöhrer, M., Zdun, U.: Foundational oracle patterns: Connecting
  blockchain to the off-chain world. Business Process Management: Blockchain
  and Robotic Process Automation Forum p. 35–51 (2020)

\bibitem{bitcoin-whitepaper}
Nakamoto, S.: {Bitcoin}: A peer-to-peer electronic cash system (2008),
  \url{https://www.tezos.com/static/papers/white-paper.pdf}

\bibitem{web3.js}
Vogelsteller, F., Kotewicz, M., Wilcke, J., Oance, M.: web3.js - {Ethereum
  JavaScript API}, \url{https://web3js.readthedocs.io/en/v1.3.4/}

\bibitem{web3.py}
Vogelsteller, F., Kotewicz, M., Wilcke, J., Oance, M.: web3.py - a {Python}
  library for interacting with {Ethereum},
  \url{https://web3py.readthedocs.io/en/stable/}

\end{thebibliography}
